\newif\ifllncs\llncsfalse
\newif\ifanon\anonfalse

\ifllncs
\documentclass{llncs}
\else
\documentclass[11pt,letter]{article}
\fi 

\usepackage{breakcites}
\usepackage[utf8]{inputenc}
\usepackage[T1]{fontenc}

\ifllncs
\else
\usepackage{beton}
\usepackage{fullpage}
\fi

\usepackage{amsmath}
\usepackage{amsfonts}
\usepackage{amssymb}
\usepackage{mathtools, xparse}
\usepackage[dvipsnames]{xcolor}
\usepackage{graphicx}
\usepackage{braket}
\usepackage{url}
\usepackage{bm}
\usepackage{soul}
\usepackage{multirow}
\definecolor{DarkBlue}{RGB}{0,0,150}
\definecolor{NotSoDarkBlue}{RGB}{15,15,210}
\definecolor{DarkRed}{RGB}{150,0,0}
\definecolor{DarkGreen}{RGB}{0,100,0}
\usepackage[pdfstartview=FitH,pdfpagemode=UseNone,colorlinks,linkcolor=DarkRed,filecolor=blue,citecolor=DarkRed,urlcolor=DarkRed,pagebackref,breaklinks]{hyperref}
\usepackage[ruled, algosection]{algorithm2e}

\usepackage{microtype}      
\usepackage{libertine}
\usepackage{libertinust1math}

\DeclareRobustCommand{\looongrightarrow}{%
  \DOTSB\relbar\joinrel\relbar\joinrel\relbar\joinrel\relbar\joinrel\relbar\joinrel\relbar\joinrel\relbar\joinrel\relbar\joinrel\relbar\joinrel\relbar\joinrel\relbar\joinrel\relbar\joinrel\relbar\joinrel\relbar\joinrel\rightarrow
}

\usepackage{orcidlink}

\newcommand{\CC}{\mathbb{C}}
\newcommand{\RR}{\mathbb{R}}

\newcommand{\NN}{\mathbb{N}}
\newcommand{\ZZ}{\mathbb{Z}}

\newcommand{\poly}{\mathsf{poly}}
\newcommand{\negl}{\mathsf{negl}}
\newcommand{\Bmax}{B_\mathrm{max}}

\DeclarePairedDelimiter{\floor}{\lfloor}{\rfloor}

\newcommand{\jac}[2]{\ensuremath{\left(\frac{#1}{#2}\right)}}

\allowdisplaybreaks

\DeclarePairedDelimiterX{\bigket}[1]{\bigg\lvert}{\bigg\rangle}{\,#1}
\DeclarePairedDelimiter{\norm}\lVert\rVert

\newcommand{\condparagraph}[1]{\ifllncs \subsubsection{#1} \else \paragraph{#1} \fi}

\ifllncs
\else
\usepackage{amsthm}
\newtheorem{theorem}{Theorem}

\newtheorem{lemma}[theorem]{Lemma}
\newtheorem{corollary}[theorem]{Corollary}
\newtheorem{definition}[theorem]{Definition}
\newtheorem{proposition}[theorem]{Proposition}
\newtheorem{remark}{Remark}

\numberwithin{theorem}{section}
\numberwithin{conjecture}{section}
\numberwithin{problem}{section}
\fi

\newif\ifnotes
\notesfalse

\title{The Jacobi Factoring Circuit:\\ {\large Quantum Factoring with Near-Linear Gates and Sublinear Space and Depth}}

\ifanon
    \author{}
    \date{}
\else
    \author{
    Gregory D. Kahanamoku-Meyer\thanks{Email: \texttt{gkm@mit.edu}. Supported by U.S. DoE Co-design Center for Quantum Advantage (C$^2$QA) DE-SC0012704.}\\MIT \and
    Seyoon Ragavan\thanks{Email: \texttt{sragavan@mit.edu}. Supported by NSF CNS-2154149, a Simons Investigator Award, the Defense Advanced Research Projects Agency (DARPA) under Contract No. HR0011-25-C-0300, and Amazon Research Awards. Any opinions, findings and conclusions or recommendations expressed in this material are those of the author(s) and do not necessarily reflect the views of the Defense Advanced Research Projects Agency (DARPA).}\\MIT 
    \and 
    Vinod Vaikuntanathan\thanks{Email: \texttt{vinodv@mit.edu}. Supported by NSF CNS-2154149 and a Simons Investigator Award.}\\MIT \and
    Katherine Van Kirk\thanks{Email: \texttt{kvankirk@g.harvard.edu}. Supported by the Fannie and John Hertz Foundation and an NDSEG fellowship.}\\Harvard
    }
\fi 

\ifllncs\else
\date{\today} \fi

\begin{document}
\maketitle

\begin{abstract}
    We present a compact quantum circuit for factoring a large class of integers, including some whose classical hardness is expected to be equivalent to RSA (but not including RSA integers themselves). Most notably, we factor $n$-bit integers of the form $P^2 Q$ with $\log Q = \Theta(n^a)$ for $a \in (2/3, 1)$ in space and depth sublinear in n (specifically, $\widetilde{O}(\log Q)$) using $\widetilde{O}(n)$ quantum gates; for these integers, no known classical algorithms exploit the relatively small size of $Q$ to run asymptotically faster than general-purpose factoring algorithms. To our knowledge, this is the first polynomial-time circuit to achieve sublinear qubit count for a classically-hard factoring problem.
 
    Our circuit builds on the quantum algorithm for squarefree decomposition discovered by Li, Peng, Du, and Suter (Nature Scientific Reports 2012), which relies on computing the Jacobi symbol in quantum superposition. The technical core of our contribution is a new space-efficient quantum algorithm to compute the Jacobi symbol of $A$ mod $B$, in the regime where $B$ is classical and much larger than $A$. Our circuit for computing the Jacobi symbol generalizes to related problems such as computing the greatest common divisor and modular inverses, and thus could be of independent interest.
\end{abstract}

\thispagestyle{empty}
\newpage
\thispagestyle{empty}
\tableofcontents
\newpage 
\pagenumbering{arabic}


\section{Introduction}

Shor's discovery of a polynomial-time quantum algorithm for factoring numbers~\cite{shor97} jump-started the field of quantum computation.
However, despite decades of intense research and development in quantum algorithms, quantum error correction, and quantum hardware, quantum factoring circuits remain out of reach for current devices.

The difficulty is many-fold; a primary issue is that that the quantum circuits for factoring are still rather large, in terms of gate count, space complexity, and gate depth. 
For example, Shor's algorithm seems to require quantum circuits of size $\tilde{O}(n^2)$ to factor $n$-bit numbers~\cite{shor97} (where the notation $\tilde{O}$ hides factors poly-logarithmic in $n$); a recent improvement by Regev has reduced the asymptotic gate count to $\tilde{O}(n^{3/2})$ per run ~\cite{Regev23}, but the practical cost of achieving this improved scaling seems rather large~\cite{DBLP:journals/corr/abs-2405-14381}. 
Much effort has been applied to reducing the space-complexity of these circuits, for both Shor~\cite{beckman, vedral, seifert2001using, Cop02, CleveW00, beauregard, takahashi, zalka2006shors, DBLP:conf/pqcrypto/EkeraH17, gidney2017factoring, hrs17, gidney2019, DBLP:journals/quantum/GidneyE21, kahanamokumeyer2024fast} and Regev's algorithms~\cite{rv24,ekeragartner}, achieving space as low as $\widetilde{O}(n)$ or even $O(n)$ qubits.
Particularly notable is a recent work which reduced the space cost of Shor's algorithm to $n/2 + o(n)$ qubits~\cite{chevignard_reducing_2024}.
Yet there seems to be no fundamental obstacle preventing these costs from being improved further, and indeed if we are to have any hope of factoring classically-intractable integers on quantum computers in the near- or medium-term, it will be necessary to do so.
Thus we arrive at two questions that are the focus of this paper: 
\begin{quote}
    \begin{center}
        {\em Are there quantum circuits for factoring with (near-)linear gate count?\\
        Could these be implemented with sub-linear space and depth?}
    \end{center}
\end{quote}
In a nutshell, our main contribution is to present the {\em Jacobi factoring circuit}, a quantum circuit for factoring a large class of integers for which efficient classical algorithms are not known.
Our circuit completely factors any number $N$ whose prime decomposition has distinct exponents, and finds at least one non-trivial factor if any exponent is $\geq 2$ (i.e. $N$ is divisible by the square of some prime). 
Notably, this excludes RSA composites $N=PQ$ that are a product of two primes.
We state below the special case of $N=P^2Q$ with $P$ and $Q$ prime. 
\begin{theorem}[Informal, see Corollary~\ref{corr:minimal_gates} for formal statement]\label{thm:informal}
   There is a quantum circuit that factors any $n$-bit integer $N=P^2 Q$ (with $P$ and $Q$ prime, and $Q<2^m$ for some $m$) with $\widetilde{O}(n)$ gates, $\widetilde{O}(m)$ qubits, and $\widetilde{O}(m + n/m)$ depth.
\end{theorem}
The space and depth complexity are sublinear when $m = \Theta(n^a)$ with $a \in (0, 1)$. Algorithms to factor numbers of this form have been extensively studied in the cryptography and computational number theory literature~\cite{Peralta1996FasterFO,DBLP:conf/crypto/BonehDH99,DBLP:conf/asiacrypt/CastagnosJLN09,DBLP:conf/eurocrypt/CastagnosL09,DBLP:series/isc/May10,DBLP:conf/ctrsa/CoronFRZ16,Harvey_2022,mulder24}; motivated in part by the fact that this problem's hardness has been used as the basis of several cryptosystems~\cite{50373,OkamotoU98,DBLP:journals/joc/PaulusT00,DBLP:conf/crypto/Takagi98,SCHMIDTSAMOA200679}. 
Roughly speaking, to classically factor $N=P^2Q$ where $Q$ is the smaller of the two numbers, we have two choices. 
Either employ a class of special-purpose factoring algorithms along the lines of Lenstra's elliptic curve method~\cite{lenstraecm}, which run in time $\exp(\tilde{O}(\sqrt{\log Q}))$; or use the fastest general-purpose factoring algorithm, namely the number field sieve~\cite{pollard,llmp,blp}, which runs in time $\exp(\tilde{O}((\log N)^{1/3}))$. Which one is faster depends on how small $Q$ is relative to $N$. As long as $\log Q = \widetilde{\Omega}((\log N)^{2/3})$, there are no known {\em classical} algorithms that exploit the special structure in $N$ to factor faster than general-purpose factoring algorithms.
We refer the reader to Section~\ref{sec:cntprelims} for more in-depth discussion on special-purpose factoring.

Before proceeding further, let us mention that the other barriers to realizing integer factorization on a quantum computer come from the concrete costs of factoring circuits; from the overhead due to quantum error-correction~\cite{DBLP:journals/quantum/GidneyE21}; and from the difficulty in building quantum hardware~\cite{martinis2019}, none of which we address in this paper. We do note, however, that the Jacobi factoring circuit, appropriately instantiated, seems friendly enough to admit a concretely small realization that factors, say, $2048$-bit integers of the form stated in Theorem~\ref{thm:informal}, although we leave an exploration of the circuit's concrete costs to future work.

The Jacobi factoring circuit builds on the remarkable, but apparently little known, work of Li, Peng, Du and Suter~\cite{LPDS12} who constructed a quantum circuit to compute the squarefree decomposition of an integer. That is, given as input a positive integer $N$, find the unique $A$ and $B$ such that $N = A^2B$ and $B$ is not divisible by the square of any integer (greater than $1$). The quantum part of the LPDS circuit computes a Jacobi symbol mod $N$, followed by a quantum Fourier transform mod $N$. 
Using the algorithm of Hales and Hallgren~\cite{DBLP:conf/focs/HalesH00}, the quantum Fourier transform mod $N$ can be computed (approximately) with near-linear size quantum circuits. We observe that using  
Sch\"{o}nhage's GCD algorithm~\cite{schonhageeuc,thull1990uni,bach1996algorithmic, Moller2008OnSA}, Jacobi symbols can be computed in near-linear time as well. Overall, this gives a near-linear size (and also near-linear space and depth) quantum circuit for squarefree decomposition which, in particular, factors numbers $N=P^2Q$ where $P,Q$ are prime numbers.

\paragraph{Our Contributions.} Building on the aforementioned work of Li, Peng, Du and Suter~\cite{LPDS12}, we show the following: 
\begin{itemize}
   
    \item Our first contribution, presented in Section~\ref{sec:highlevel}, is a new analysis of \cite{LPDS12} that is \emph{necessary} to get our final result. Jumping ahead a bit, it allows the quantum circuit to use a superposition of numbers from a potentially much smaller range, e.g. to factor $N=P^2Q$, one can use a superposition of numbers from 1 to $\poly(Q)$ rather than 1 to $\poly(N)$ as in \cite{LPDS12}. (Hence the number of qubits required for the initial superposition will be $\log (\poly(Q)) = O(\log Q)$.)
    
    \item Our second and main technical contribution, presented in Section~\ref{sec:jacobi}, is the construction of an efficient quantum circuit to compute the Jacobi symbol $\jac{A}{B}$ in near-linear size and sublinear space and depth, when $\log A \ll \log B$, and $A$ could be in superposition but $B$ is classical. In particular, our circuit achieves qubit count $\widetilde{O}(\log A)$ and gate count $\widetilde{O}(\log B)$, parallelized into depth at most $\widetilde{O}(\log B/\log A + \log A)$. Combined with our first contribution, this yields a factoring circuit for $P^2Q$ with essentially the same gate count as \cite{LPDS12} but using smaller space and depth when $\log Q \ll \log P$.
    
    We believe our circuit design is of independent interest as it can be readily adapted to solve other problems of a similar nature, e.g. computing the greatest common divisor $\gcd(A, B)$ with the same efficiency.
    
     \item Our final contribution, presented in Section~\ref{sec:fullyfactoring}, is the observation that an algorithm for squarefree decomposition suffices to completely factor a general class of integers of inverse-polynomial density, namely any integer whose prime factorization has distinct exponents.
     A similar observation is well-known in the context of factoring polynomials~\cite{yun}, and the high-level ideas are similar between the two settings.
\end{itemize}

\paragraph{On Special-Purpose Classical and Quantum Factoring.} 
Our result can be seen to complement the classical factoring algorithms, e.g.~\cite{lenstraecm,mulder24} that exploit various types of structure. The result by~\cite{LPDS12} takes a first step in this direction by showing that some integers with special structure (e.g. $N = P^2Q$ with $P, Q$ prime) become polynomially easier to factor quantumly. Our result builds on this, demonstrating that these integers can be quantumly factored even more easily (i.e. in much lower space and depth) if $\log Q \ll \log N$. We emphasize that, as long as $\log Q \geq \widetilde{\Omega}((\log N)^{2/3})$, this structure cannot be exploited by any known special-purpose classical factoring algorithms (we discuss this more in Section~\ref{sec:cntprelims}).

In contrast, prior quantum factoring algorithms~\cite{shor97,Regev23} do not seem to benefit \emph{polynomially} from any such structure. 
Remarkable algorithms by~\cite{DBLP:conf/pqcrypto/EkeraH17, chevignard_reducing_2024} that achieve constant-factor improvements over Shor's original construction depend upon the input being an RSA integer (the product of two primes of equal bit length), but no prior algorithms have better asymptotic scaling depending on the structure of the input.
A related, and important, open question is whether our algorithm can be leveraged or extended to factor integers in general (we note that it would suffice to devise a way to factor squarefree integers; see Sections~\ref{sec:highlevel} and~\ref{sec:fullyfactoring} for details).

Another important future direction is classical cryptanalysis for factoring integers of the form $P^2Q$ where $Q$ is much smaller than $P$, since this is the regime in which we get sublinear space. To the best of our knowledge, existing algorithms~\cite{lenstraecm, DBLP:conf/crypto/BonehDH99, mulder24} do not offer any significant improvements in this regime, but this is not a regime that was previously of much practical interest, so we leave further investigation along these lines to future work.

\paragraph{A More Efficient Proof of Quantumness.} Recent excitement has centered on \emph{efficiently-verifiable proofs of quantumness}, which are protocols by which a single untrusted quantum device can demonstrate its quantum capability to a skeptical polynomial-time classical verifier~\cite{brakerski_cryptographic_2021, brakerski_simpler_2020, KCVY21, yamakawa_verifiable_2022, kalai_quantum_2023, morimae_proofs_2023, alnawakhtha_lattice-based_2024, aaronson_verifiable_2024, miller_hidden-state_2024}. 
The Jacobi factoring circuit presented in this work immediately yields the first factoring-based proof of quantumness with sublinear space complexity (see Table~\ref{tab:poq_costs}).
Existing proofs of quantumness based on factoring broadly fall into two categories: factoring algorithms, which straightforwardly demonstrate their quantum capability by finding the factors; and interactive protocols, which do not actually factor the number, but instead perform a task that for any classical algorithm is provably as hard as factoring.
We address each of these in turn:
\begin{itemize}
    \item \textit{Factoring algorithms:} Shor's algorithm for factoring~\cite{shor97}, when implemented with a low-depth quantum multiplication circuit~\cite{nie_quantum_2023}, costs $\widetilde{O}(n^2)$ gates, $\widetilde{O}(n)$ qubits, and $\widetilde{O}(n)$ depth.\footnote{There also exist log-depth implementations of Shor's algorithm~\cite{CleveW00}, but they come at the cost of far worse gate and qubit counts. We include asymptotics for this circuit in Table~\ref{tab:poq_costs} for completeness.} Regev's recent improved factoring algorithm~\cite{Regev23}, together with the optimizations of~\cite{rv24} (and using the same low-depth multiplier), can be implemented in $\widetilde{O}(n^{1.5})$ gates, $\widetilde{O}(n)$ qubits, and $\widetilde{O}(n^{0.5})$ depth. 
    Note, however, that this is the gate count per run, and as $O(n^{0.5})$ runs are required, Regev does not improve total gate count across runs.
    The previously proposed Jacobi factoring circuit~\cite{LPDS12}, together with the algorithm by~\cite{schonhageeuc} for computing Jacobi symbols, uses $\widetilde{O}(n)$ gates, space, and depth.

    In contrast, if we instantiate our construction with an integer $N = P^2Q$ where $\log Q = \widetilde{\Theta}((\log N)^{2/3})$, our circuit uses $\widetilde{O}(n)$ gates, $\widetilde{O}(n^{2/3})$ depth, and $\widetilde{O}(n^{2/3})$ qubits.
    In terms of the product of qubit count with either gates or depth, this outperforms all other factoring algorithms described here.

    \item \emph{Interactive protocols} that do not factor: the relevant protocols here are those based on trapdoor claw-free functions (TCFs), specifically instantiated with Rabin's function $f(x) = x^2 \bmod{N}$ as introduced in~\cite{KCVY21}. 
    Evaluating this function requires performing just a single multiplication, and thus can be implemented with a quantum circuit of $\widetilde{O}(n)$ gates, $\mathsf{polylog}(n)$ depth, and $\widetilde{O}(n)$ qubits. 
    While the low depth is appealing, the main obstacle is the qubit count, which is outperformed substantially by the Jacobi factoring circuit. 
    Furthermore, the protocol of~\cite{KCVY21} is interactive, requiring the quantum computer to maintain coherence throughout several rounds of measurement of subsets of the qubits and communication with the verifier.
    Indeed, that protocol cannot be run in an ``offline'' setting, where a classical verifier publishes a challenge publicly and provides no further data to any particular prover.
    Interactive protocols are also somewhat less satisfying as a proof of quantum computational power, because the prover is not actually solving a computational problem---instead, interaction allows the prover show it can make measurements in anticommuting bases, which is not possible for a classical algorithm.
    There do exist TCF-based protocols which are non-interactive; their classical hardness either relies on quantum access to random oracles (\cite{brakerski_simpler_2020}, see \cite{canetti_random_2004, koblitz_random_2015} for discussion of the random oracle heuristic) or computational problems other than factoring~\cite{arabadjieva_single-round_2024}.
\end{itemize}

For completeness we note that, when the quantum circuit costs are expressed as a function of the best-known classical time cost for the same problem, our result does not asymptotically outperform certain proofs of quantumness based the hardness of problems other than factoring --- simply because the classical hardness of those problems grows much more rapidly.
Consider, for example, applying Shor's algorithm to the elliptic curve discrete logarithm problem (ECDLP)~\cite{shor97, haner_improved_2020}.
Although the standard quantum circuit to solve ECDLP requires at least linear space and depth in the size of the input --- which is worse than the circuits we present in this work --- this is outweighed by the fact that integer factorization admits sub-exponential time classical algorithms (namely $\exp(\widetilde{O}(n^{1/3}))$), while to the best of our knowledge ECDLP does not. 
Thus, if we want to work with a problem that takes time $T$ to solve classically, it would suffice to set $n = O(\log T)$ in the case of ECDLP, whereas for factoring we would need to set $n = \widetilde{O}((\log T)^3)$.
However, in practice the constant factors for factoring circuits seem to be dramatically better than those for the ECDLP problem, with the constant multiplying the leading-order term even being less than 1 in some cases~\cite{DBLP:journals/quantum/GidneyE21}.

\begin{table}
\begin{center}
    \begin{tabular}{|c|c|c|c|}
    \hline
    \multirow{2}{*}{Protocol} & \multicolumn{3}{|c|}{Cost (up to polylog factors)} \\
    \cline{2-4}
    & Gates & Depth & Qubits \\
    \hline\hline
    Shor~\cite{shor97} & $n^{2}$ & $n$ & $n$ \\
    \hline
    Log-depth Shor~\cite{CleveW00} & $n^{5}$ & $\log n$ & $n^{5}$ \\
    \hline
    Regev~\cite{Regev23, rv24} & $n^{3/2}$ & $n^{1/2}$ & $n$ \\
    \hline
    $x^2 \bmod{N}$~\cite{KCVY21}$^\dagger$ & $n$ & $\log^2 n$ & $n$ \\
    \hline
    Squarefree decomposition & \multirow{2}{*}{$n$} & \multirow{2}{*}{$n$} & \multirow{2}{*}{$n$} \\
    \cite{LPDS12, schonhageeuc} & & & \\
    \hline
    \textbf{This work} & $\bm{n}$ & $\bm{n^{2/3}}$ & $\bm{n^{2/3}}$ \\
    \hline
    \end{tabular}
    \caption{\textbf{Asymptotic cost of various proofs of quantumness based on the hardness of factoring $n$-bit integers.}
    We omit constant and poly-logarithmic factors throughout for clarity.
    For all algorithms which use black-box multiplication, we assume the use of a parallelized circuit for Sch\"onhage-Strassen multiplication~\cite{nie_quantum_2023}.
    We use $\dagger$ to denote the fact that~\cite{KCVY21} is an interactive protocol in which the quantum computer is not required to actually factor the number.}
    \label{tab:poq_costs}
\end{center}
\end{table}

\medskip
\noindent
We proceed to describe the quantum circuit of \cite{LPDS12} and then our technical contributions in more detail.

\subsection{The LPDS Circuit for Squarefree Decomposition}\label{sec:lpdsoverview}

In a beautiful work from a decade ago, Li, Peng, Du and Suter~\cite{LPDS12} showed a quantum circuit to compute the squarefree decomposition of an integer. That is, let $N=A^2B$ where $B$ is not divisible by the square of any integer (greater than 1) denote the unique squarefree decomposition of $N$. Given $N$, computing $B$ seems classically hard in general; indeed, it is at least as hard as factoring integers of the form $N=P^2Q$ where $P$ and $Q$ are primes. The squarefree decomposition problem  has received much attention from the computational number-theory community~\cite{Peralta1996FasterFO,DBLP:conf/crypto/BonehDH99,DBLP:conf/asiacrypt/CastagnosJLN09,DBLP:conf/eurocrypt/CastagnosL09,DBLP:series/isc/May10,DBLP:conf/ctrsa/CoronFRZ16,Harvey_2022,mulder24}, in part due to its applications in cryptography~\cite{50373,OkamotoU98,DBLP:journals/joc/PaulusT00,DBLP:conf/crypto/Takagi98,SCHMIDTSAMOA200679}, and it is at the core of other important problems such as computing the ring of integers of a number field~\cite{BuchmannLenstra} and the endomorphism ring of an elliptic curve over a finite field~\cite{BISSON2011815}.

The starting point of \cite{LPDS12} is the observation that when $N=P^2Q$, the Jacobi symbol of $x$ mod $N$ depends essentially\footnote{We say ``essentially'' because this is subject to the minor constraint that $x$ and $N$ need to be relatively prime.} only on $x$ mod $Q$. Indeed, if $x$ and $N$ are relatively prime, 
$$ \left(\frac{x}{N}\right) = \left(\frac{x}{P}\right)^2 \left(\frac{x}{Q}\right) = \left(\frac{x}{Q}\right)  $$
since $\left(\frac{x}{P}\right) \in \{\pm 1\}$. Thus, the Jacobi symbol of $x$ mod $N$ is periodic modulo the secret factor $Q$.

With quantum period finding in mind, this naturally suggests the following procedure: (1) start with a uniform superposition over all $x$ mod $N$; (2) compute and measure the Jacobi symbol $\jac{x}{N}$; and (3) use Shor's period-finding procedure~\cite{shor97} to recover $Q$.
The apparent obstacle is that once we apply a phase of $\jac{x}{N}$, we will end up not with one periodic signal modulo $Q$ but a \emph{superposition} of several periodic signals modulo $Q$.

One approach to circumvent this obstacle would be the following: instead of measuring one Jacobi symbol, we could measure multiple Jacobi symbols $\jac{x}{N},\jac{x+1}{N},\ldots,\jac{x+k}{N}$ for a large enough $k$ so as to (hopefully) uniquely determine the value of $x$ mod $Q$. (Intuitively, each one of these Jacobi symbols should give an ``independent'' piece of information about the value of $x$ mod $Q$, so measuring enough of them should determine $x$ mod $Q$.) It turns out that $k=\mathsf{poly}(\log Q)$ likely suffices (see the Boneh-Lipton conjecture~\cite{DBLP:conf/crypto/BonehL96,DBLP:conf/crypto/CorriganGibbsW24}). Thus measuring the function $\mathsf{manyJac}_{N,k}(x) = \left(\frac{x}{N}\right), \left(\frac{x+1}{N}\right), \ldots, \left(\frac{x+k}{N}\right)$ on the uniform superposition gives us
\begin{equation}
\label{LPDS}
\frac{1}{\sqrt{N}} \cdot \sum_{x \in [0,N-1]} \ket{x} 
\stackrel{\mbox{\small measure $\mathsf{manyJac}_{N,k}$}}{\looongrightarrow}
\frac{1}{P} \cdot \sum_{j \in [0,P^2-1]} \ket{x_0+jQ} 
\stackrel{\mbox{\small QFT mod $N$}}{\longrightarrow}
\frac{1}{\sqrt{Q}} \sum_{j \in [0, Q-1]} e^{-2\pi i x_0 j /Q} \bigket{\frac{jN}{Q}}  
\end{equation}
for some $x_0 \in [0,Q-1]$. Now, measuring gives us an integer multiple of $N/Q = P^2$ from which it is not hard to read off $P^2$ and therefore $P$.

This would, however, result in a rather large circuit: to uniquely fix $x$ mod $Q$, one would certainly need to compute at least $\Omega(\log Q)$ Jacobi symbols (and perhaps even a larger $\poly(\log Q)$~\cite{DBLP:conf/crypto/BonehL96,DBLP:conf/crypto/CorriganGibbsW24}). The key result of \cite{LPDS12} is that just computing and measuring a single Jacobi symbol already suffices (even though we would be working with a superposition of periodic signals\footnote{\cite{DBLP:conf/focs/HalesH00} provides a black-box algorithm for finding the period of ``many-to-one'' periodic functions like this, however it requires a super-constant number of calls to the Fourier sampling subroutine. The Gauss sum analysis of \cite{LPDS12} (and that of the present work) provides much better efficiency, showing that just one iteration of Fourier sampling suffices to find the period with probability $\Omega(1)$.}). Even better, instead of measuring the Jacobi symbol we can simply apply a phase equal to the Jacobi symbol $\jac{x}{N}$. The effect of this is to eliminate any amplitude placed by the post-QFT state on $\ket{0}$ (which is not helpful for factoring $N$).
Indeed, the QFT of the signal after measuring a single Jacobi symbol, namely that of $x$ mod $N$, will be very similar to the end result in Equation~\eqref{LPDS} except that each basis state will receive a sum of several amplitudes. In particular, if we apply a phase equal to the Jacobi symbol $\jac{x}{N}$\footnote{We once again assume here that the Jacobi symbol is always in $\left\{-1, 1\right\}$ even though it can occasionally also be 0; we will handle this more carefully in the relevant technical sections.} each non-zero basis state $\ket{jN/Q}$ will have absolute amplitude
$$ \approx \frac{1}{Q} \cdot \left| \sum_{x_0 \in [1, Q-1]} \jac{x_0}{Q}\exp\left(-\frac{2\pi i x_0 j}{Q}\right) \right|.$$
By standard Gauss sum bounds (see Section~\ref{sec:gausssums} and Remark~\ref{remark:gausssums} for details), the summation is {\em lower-bounded} by $\Omega(\sqrt{Q})$, and so we know that each non-zero basis state will have amplitude $\Omega(1/\sqrt{Q})$. Since there are $Q-1$ such states, a measurement will give us a non-zero multiple of $N/Q=P^2$ with a constant probability (in fact, \cite{LPDS12} shows that this probability is 1!).

We remark that \cite{LPDS12} generalizes this method to obtain the squarefree decomposition of any $N$, not necessarily of the form $P^2Q$ for prime $P, Q$. With this in mind, we now turn to an overview of our techniques.

\subsection{Technical Overview}\label{sec:techoverview}

\paragraph{Section~\ref{sec:highlevel}: A New Analysis of~\cite{LPDS12}.} Our first contribution is a more careful analysis of the circuit by~\cite{LPDS12}, wherein we show that it suffices to start with a superposition from $1$ to $\poly(Q)$ rather than all the way to $N$. At a high level, this follows from combining two previous techniques. Our starting point is the analysis by Shor~\cite{shor97} that shows that in order to find the period of a function with period $\leq Q_\mathrm{max}$, it suffices to take a superposition from $1$ to $\poly(Q_\mathrm{max})$.

The reason this does not immediately suffice for our setting is that we are not working with one periodic signal; we would be working with a superposition of periodic signals corresponding to values $x_0 \in [0, Q-1]$ grouped according to their Jacobi symbol $\jac{x_0}{Q}$. To get around this, we combine elements of Shor's analysis~\cite{shor97} with the Gauss sum analysis introduced by~\cite{LPDS12}.

\paragraph{Section~\ref{sec:jacobi}: Computing Jacobi Symbols in Sublinear Space and Depth.} The computational bottleneck in the~\cite{LPDS12} factoring circuit is computing the Jacobi symbol $\jac{x}{N}$, where $x \in [0, N-1]$ is in superposition. This can be done in gates (and hence space/depth) $\widetilde{O}(n)$~\cite{schonhageeuc, thull1990uni,bach1996algorithmic, Moller2008OnSA}, where $n$ is the number of bits in $N$, which is a near-linear gate complexity and hence essentially tight. 

However, thanks to our first contribution, we need only compute $\jac{x}{N}$ for $x \leq \poly(Q)$; moreover, since $N$ is classically known, there is the tantalizing possibility that the number of qubits could be pushed down to linear in $\log Q$ rather than $\log N$.
We show that this is indeed the case (up to polylogarithmic factors), by constructing a quantum circuit that computes the Jacobi symbol using space $\widetilde{O}(m)$ qubits for any $m \geq \log Q$.
Our circuit is also very efficient, achieving gate count $\widetilde{O}(n)$, parallelized into depth $\widetilde{O}(n/m + m)$.\footnote{If $\log Q < O((\log N)^{1/2})$, the depth and space cannot \textit{both} be made to scale with $\log Q$ simultaneously, because the space-time volume (space times depth) is lower bounded by the gate count, and the gate count is lower bounded by $O(\log N)$. However, the parameters can be tuned to achieve a continuous tradeoff between the two, while maintaining a space-time product nearly linear in $\log N$.}

To explain our methods, let us revisit some well-known algorithms for computing the Jacobi symbol $\jac{x}{N}$. 
These algorithms also provide algorithms for computing GCDs and vice versa:
\begin{itemize}
    \item The binary GCD algorithm~\cite{bach1996algorithmic} is often used in quantum algorithms due to its circuit-friendliness \cite{roetteler17}. However, this will not be useful for our goals; the number of gates needed to compute $\jac{x}{N}$ is $O((\log x + \log N)^2)$~\cite{bach1996algorithmic}, which is quadratic in $\log N$ (rather than near-linear).
    \item The extended Euclidean algorithm relies on the observation that the Jacobi symbol $\jac{a}{b}$ is equal to $\jac{a \bmod{b}}{b}$, which together with quadratic reciprocity (property 7 of Theorem~\ref{thm:jacobiproperties}) allows one to rapidly reduce the size of the problem's inputs. Indeed, after just one step, the problem is reduced to the computation of the Jacobi symbol of two inputs of length $O(\log x)$. Nevertheless, due to that first step, this algorithm seems to require $\widetilde{O}(\log N)$ qubits.
    \item Finally, there is an algorithm due to Sch{\"o}nhage~\cite{schonhageeuc,thull1990uni, bach1996algorithmic, Moller2008OnSA} that runs in $\widetilde{O}(\log N)$ gates, but does not come with any better guarantees on the space and depth.
\end{itemize}
We take an approach that, at a very high level, mimicks the extended Euclidean algorithm:
\begin{enumerate}
    \item\label{item:techoverviewstep1} First, we reduce the computation of $\jac{x}{N}$ to some Jacobi computation $\jac{a}{b}$ between two inputs $a, b$ of length $O(\log x)$.
    \item\label{item:techoverviewstep2} We then compute $\jac{a}{b}$ using Sch{\"o}nhage's algorithm out-of-the-box, which only requires gates (and hence space/depth) $\widetilde{O}(\log x)$.
\end{enumerate}
The challenge, and room for creativity, is in implementing step~\ref{item:techoverviewstep1}.
To do this, we find a multiple $kx$ of $x$ such that both of the following are true: (a) $N - kx$ is divisible by $2^{n-m}$; and (b) $kx < 2^n$. 
For now one should consider $m = \lceil \log x \rceil$; in some cases, one may choose $m \geq \lceil \log x \rceil$ to improve efficiency, as we discuss later. 
To be explicit, this allows us to compute the Jacobi symbol via the following chain of transformations:
$$\jac{x}{N} \rightarrow \jac{N}{x} \rightarrow \jac{N-kx}{x} \rightarrow \jac{(N - kx)/2^{n-m}}{x},$$
where each of the arrows follows from standard properties of the Jacobi symbol stated in Theorem~\ref{thm:jacobiproperties}, and the Jacobi symbol of the last expression is computed directly (step~\ref{item:techoverviewstep2} above).

A conceptually simpler but concretely less direct and efficient variant\footnote{We thank Daniel J. Bernstein for suggesting this perspective on our algorithm.} of our algorithm essentially follows the ``extended Euclidean'' blueprint: we use $kx$ to compute $N \bmod{x} = \left[\left(\frac{N-kx}{2^{n-m}}\right) \cdot \left(2^{n-m} \bmod{x}\right)\right] \bmod{x}$. Then we can transform $\jac{N}{x} \rightarrow \jac{N \bmod{x}}{x}$ and finish from there using~\cite{schonhageeuc}. In practice, this is unnecessary extra computation, so we proceed using our more specialized blueprint and do not bother with computing $N \bmod{x}$. One could of course attempt to directly compute $N\bmod{x}$ as usual using long division, but it is not clear how to do this reversibly in low space and depth.

The key idea behind our quantum circuit for step~\ref{item:techoverviewstep1} is to stream through the (classical) bits of $N$ in blocks of size $m$ starting with the lowest-order bits, matching each block of $kx$ to the corresponding block of $N$ (such that the difference $N-kx$ has trailing zeros).
Sublinear quantum space is achieved via the observation that only the leading $O(m)$ bits of the running sum $kx$ need to be stored quantumly at any given time, as all of the lower-order bits match the classical bits of $N$ by design.
Sublinear depth follows from the fact that the number of blocks is $O(n/m)$, and the desired operations on each block can be performed in a constant number of multiplications of depth $O(\operatorname{polylog}(m))$~\cite{nie_quantum_2023, schonhage1971fast}.

Our algorithm builds on Montgomery reduction~\cite{Montgomery1985} and the binary GCD algorithm and can be thought of as a ``reversed'' variant of long division; long division starts from the most significant bit (MSB) and iterates towards the least significant bit (LSB) to find a multiple of $x$ that agrees with $N$ in the MSBs, whereas we start from the LSB and iterate towards the MSB to find a multiple of $x$ that agrees with $N$ in the LSBs. The benefit of proceeding in this reversed way is that uncomputing intermediate states now becomes easy, by way of a simple comparison that only depends on the MSBs of our intermediate state.

\paragraph{Section~\ref{sec:fullyfactoring}: Completely Factoring Special Integers.} Finally, we present a black-box reduction implying that any algorithm for squarefree decomposition can be used to \emph{completely factor} integers $N$ with distinct exponents in their prime factorization; i.e. $N$ that can be written as $p_1^{\alpha_1} \ldots p_r^{\alpha_r}$ for distinct primes $p_1, \ldots, p_r$ and distinct positive exponents $\alpha_1, \ldots, \alpha_r$. Such integers have been studied before and are referred to as \emph{special integers}~\cite{Akta2017OnTN} --- in fact, these have even been proposed for use in cryptographic applications~\cite{DBLP:journals/iacr/Schanck18}.

A similar observation is well-known in the context of factoring polynomials; the motivation for this is that if a univariate polynomial $f$ is divisible by the square of some (non-constant) polynomial, this can be detected easily by taking the GCD of $f$ and its derivative $f'$. Yun~\cite{yun} shows that this can be extended to decompose any polynomial $f(x)$ into a factorization $f(x) = g_1(x)^{\alpha_1}\ldots g_r(x)^{\alpha_r}$ where the $g_i$'s are squarefree and pairwise coprime and the $\alpha_i$'s are distinct. Many polynomial factorization algorithms thus begin with this step as a subroutine.

Unsurprisingly, the algorithm in our setting bears some high-level similarity to Yun's algorithm~\cite{yun} and is based on a simple idea: using the algorithm for squarefree decomposition, we can recover 
$$B = \prod_{i \in [r]: \alpha_i\text{ odd}} p_i.$$
Then let $i^* \in [r]$ be the index such that $\alpha_{i^*}$ is the smallest of the odd $\alpha_i$. 
(If the $\alpha_i$ are all even, then $N$ will be a perfect square and we can take its square root until at least one $\alpha_i$ is odd.)
By dividing $N$ by $B$ as many times as possible, we obtain:
$$k = \prod_{i \in [r]: \alpha_i\text{ even}} p_i^{\alpha_i} \cdot \prod_{i \in [r]: \alpha_i\text{ odd}} p_i^{\alpha_i - \alpha_{i^*}}.$$
Now, because the $\alpha_i$ are all distinct, $k$ will be divisible by every prime dividing $B$ except $p_{i^*}$. Thus we can compute $B/\gcd(k, B) = p_{i^*}$, which is a prime divisor of $N$. We can now divide as many factors of $p_{i^*}$ from $N$ as possible, then recurse.

\section{Preliminaries}

\subsection{Notation}

Let $N < 2^n$ be an $n$-bit number that we wish to factor. We use $\negl(n)$ to denote any real-valued function $f(n)$ such that $|f(n)| = o(n^{-c})$ for all constants $c > 0$.
For any positive integer $k$, let $\varphi(k)$ denote the number of positive integers in $[1, k]$ that are relatively prime to $k$. We will sometimes use the notation $A \mid B$ to indicate that the integer $A$ divides the integer $B$.

We say that an integer $B$ is \emph{squarefree} if it is not divisible by any square (other than 1). Observe that any $N$ has a unique representation of the form $A^2B$ for some squarefree $B$; indeed, if $N$ has prime factorization $\prod_{i = 1}^r p_i^{a_i}$, then we must have $B = \prod_{i = 1}^r p_i^{a_i \bmod{2}}$ and $A = \prod_{i = 1}^r p_i^{\lfloor a_i/2 \rfloor}$. When $A > 1$, we say that $N$ is \emph{squarefull}.

Throughout this paper, $\log$ will denote the base-$2$ logarithm. We use $\ZZ_N$ to denote the ring of integers mod $N$, and $\ZZ_N^*$ to denote the multiplicative group of invertible elements mod $N$. We will also use the following straightforward claim:

\begin{proposition}\label{prop:countingmodsols}
    Let $M, B$ be positive integers with $B > 1$ and let $j \in [0, B-1]$ be an integer. Then the number of integers $x \in [1, M]$ such that $x \equiv j \pmod{B}$ is exactly $$\left\lfloor \frac{M-j}{B}\right\rfloor - \left\lceil \frac{1-j}{B}\right\rceil + 1.$$
\end{proposition}
\begin{proof}
    Writing $x = By + j$, we wish to find the number of integers $y$ (not necessarily positive) such that:
    \begin{align*}
        By + j &\in [1, M] \\
        \Leftrightarrow y &\in \left[\left\lceil \frac{1-j}{B}\right\rceil, \left\lfloor \frac{M-j}{B}\right\rfloor\right].
    \end{align*}
    The conclusion now follows.
\end{proof}

\subsection{The Jacobi Symbol}

Here, we define the Jacobi symbol and state its relevant properties for our purposes. We follow the exposition in~\cite[Chapter 5]{bach1996algorithmic}. The Legendre symbol is a well-known special case of the Jacobi symbol and our starting point:

\begin{definition}[Legendre symbol]
    For an integer $a$ and an odd prime $p$, define the Legendre symbol $\left(\frac{a}{p}\right)$ as follows:
    $$\left(\frac{a}{p}\right) = \begin{cases}
0,\text{ if $p$ divides $a$;}\\
1,\text{ if $a$ is a (nonzero) quadratic residue mod $p$;}\\
-1,\text{ otherwise.}
\end{cases}$$
\end{definition}

\noindent
The Jacobi symbol is most naturally defined in terms of the Legendre symbol:

\begin{definition}[Jacobi symbol]\label{def:jacobi}
    Let $a$ be an integer and $b$ an odd positive integer with factorization $b = p_1^{e_1}\ldots p_k^{e_k}$. Then define the Jacobi symbol $\left(\frac{a}{b}\right)$ as follows:
    $$\left(\frac{a}{b}\right) = \left(\frac{a}{p_1}\right)^{e_1}\left(\frac{a}{p_2}\right)^{e_2}\ldots\left(\frac{a}{p_k}\right)^{e_k}.$$
    At various points, we will use the notation $j_b(a) \in \left\{-1, 1\right\}$ to denote a value that is guaranteed to be $\jac{a}{b}$ when $\gcd(a, b) = 1$ and can be arbitrary otherwise (since as we will see below, in this case the Jacobi symbol would be 0).
\end{definition}

The Jacobi symbol $\left(\frac{a}{b}\right)$ can be computed efficiently without knowing the factorization of $b$, via the following properties.
They can be applied, for example, in the same manner as the extended Euclidean algorithm for the greatest common divisor; we will discuss quantum circuits for computing the Jacobi symbol in detail in Section \ref{sec:jacobi}.

\begin{theorem}[Jacobi symbol properties]\label{thm:jacobiproperties}
    The Jacobi symbol has the following properties. (Recall that $\jac{a}{n}$ is only defined when $n$ is an odd positive integer --- although $a$ could be even; thus in all of the below, it is assumed that $m, n$ are both odd and positive.)
    \begin{enumerate}
        \item\label{item:jacprop1} If $\gcd(a, n) > 1$, then $\jac{a}{n} = 0$. Otherwise, $\jac{a}{n} \in \left\{-1, 1\right\}$.
        \item\label{item:jacprop2} $\jac{ab}{n} = \jac{a}{n}\jac{b}{n}$;
        \item\label{item:jacprop3} $\jac{a}{mn} = \jac{a}{m} \jac{a}{n}$;
        \item\label{item:jacprop4} $\jac{a}{n} = \jac{b}{n}$ if $a \equiv b \bmod{n}$;
        \item\label{item:jacprop5} $\jac{-1}{n} = (-1)^{(n-1)/2}$;
        \item\label{item:jacprop6} $\jac{2}{n} = (-1)^{(n^2-1)/8}$
        \item\label{item:jacprop7} (Quadratic Reciprocity) If $\gcd(a, n) = 1$ and $a$ is odd and positive, then $$\jac{a}{n} \jac{n}{a} = (-1)^{(a-1)(n-1)/4}.$$
        Consequently, whenever $a$ is odd and positive (perhaps having common factors with $n$), we will have $$\jac{a}{n} = (-1)^{(a-1)(n-1)/4} \jac{n}{a}.$$(In the case that $a, n$ have common factors, both sides will be 0.)
    \end{enumerate}
\end{theorem}

\begin{corollary}\label{cor:jacobicount}
    Let $\sigma \in \left\{-1, 1\right\}$. If $N$ is odd and not a square, the number of integers $a \in [1, N]$ such that $\jac{a}{N} = \sigma$ is $\varphi(N)/2$.
\end{corollary}
\begin{proof}
    Since $N$ is not a square, there exists a prime $p_0$ and odd integer $e_0$ such that $p_0^{e_0}$ divides $N$ but $p_0^{e_0+1}$ does not. Let the other primes dividing $N$ be $p_1, p_2, \ldots, p_k$. Then let $b \in [1, N]$ be such that $b$ is a quadratic non-residue mod $p_0$, and $b \equiv 1 \bmod{p_i}$ for all $i > 0$. 
    Such $b$ exists by the Chinese Remainder Theorem, and moreover we have for some exponents $e_0, \ldots, e_k$ that $\jac{b}{N} = \jac{b}{p_0}^{e_0} \cdot \prod_{i = 1}^k \jac{b}{p_i}^{e_i} = (-1)^{e_0} = -1$.

    Then by property~\ref{item:jacprop2} of Theorem~\ref{thm:jacobiproperties}, $a \mapsto ab\bmod{N}$ provides a bijection between the elements of $\left\{x \in \ZZ_N^*: \jac{x}{N} = 1\right\}$ and $\left\{x \in \ZZ_N^*: \jac{x}{N} = -1\right\}$. It follows that these two sets have equal cardinality. They are disjoint, and in total they comprise $\varphi(N)$ elements by property~\ref{item:jacprop1} of Theorem~\ref{thm:jacobiproperties}, so the conclusion follows.
\end{proof}

\subsection{Computational Number Theory}\label{sec:cntprelims}

We survey the classical and quantum complexity of various computational number theory problems that are relevant to this work. Recall the well-known result that if we have a classical circuit that uses $G$ gates to compute a function $f(x)$ of an input $x$, we can implement unitary computing $\ket{x}\ket{0} \mapsto \ket{x} \ket{f(x)}$ in $O(G)$ gates and $O(G)$ ancilla qubits~\cite{bennett_logical_1973}.

\paragraph{Arithmetic Operations.} The fastest known classical circuits for $n$-bit integer multiplication use $O(n \log n)$ gates~\cite{Harvey21}, and can be made quantum through standard reversibility techniques~\cite{bennett_logical_1973, bennett_timespace_1989, levine_note_1990}. 
If space is a concern, one can use the multiplier due to~\cite{kahanamokumeyer2024fast} which uses no ancilla qubits (i.e. it operates entirely in-place on the input and output registers), and has $O_\epsilon(n^{1+\epsilon})$ gates for any pre-specified $\epsilon > 0$.
The depth of multiplication can be reduced to $O(\log^2 n)$ with the use of $\widetilde{O}(n)$ ancilla qubits, via a parallel quantum circuit for the Sch\"onhage-Strassen algorithm which has gate count $\widetilde{O}(n)$~\cite{nie_quantum_2023, schonhage1971fast}.
Via Newton iteration it is possible to perform division with the same complexity as multiplication (up to constant factors in the gate count and space, and a logarithmic factor in the depth)~\cite{knuth_art_1998}.

\paragraph{Algorithms for Computing Jacobi Symbols and GCDs.} The best-known algorithms for computing Jacobi symbols and GCDs of two $n$-bit integers are:
\begin{itemize}
    \item The extended Euclidean algorithm, which can be classically done in $O(n^2)$ gates~\cite{bach1996algorithmic}. Moreover, this can be done quantumly in $O(n^2)$ gates while also keeping the space down to $O(n)$ qubits~\cite{ProosZalka03}.
    \item The binary GCD algorithm, which has the same asymptotic complexities as extended Euclidean both classically and quantumly~\cite{bach1996algorithmic, ProosZalka03}.
\end{itemize}
However, there is a faster divide-and-conquer algorithm which was conceived by Sch\"{o}nhage~\cite{schonhageeuc} and~\cite[solution to exercise 5.52]{bach1996algorithmic}, and expounded Thull and Yap~\cite{thull1990uni} and M{\"o}ller~\cite{Moller2008OnSA}. This algorithm runs in $\widetilde{O}(n)$ gates\footnote{While these algorithms are usually formulated in the Turing machine model, they can be readily transformed into circuits at the expense of multiplicative $\mathsf{polylog}(n)$ overheads~\cite{DBLP:journals/jacm/PippengerF79}.} (and hence at most that much space and depth).

\paragraph{Classical Algorithms for Factoring.} The best-known classical algorithm for factoring arbitrary $n$-bit integers is the general number field sieve~\cite{pollard, llmp, blp}, which runs in heuristic time $$\exp\left(O(n^{1/3} (\log n)^{2/3})\right).$$
However, there has also been extensive research towards generating faster classical algorithms, which exploit specific number theoretic structure present in the integer $N < 2^n$ being factored. For example, if $p < 2^{m'}$ is the smallest divisor of $N$, Lenstra's elliptic curve method~\cite{lenstraecm} recovers $p$ in heuristic time:
$$\poly(n) \cdot \exp\left(O\left(\sqrt{m' \log m'}\right)\right).$$
Furthermore, Mulder~\cite{mulder24} presents a classical algorithm specifically for squarefree decomposition, targeting the same structure as we do with the quantum algorithm in the present work.
If $N = A^2B$ with $B < 2^m$ squarefree, then Mulder's algorithm can recover $A, B$ in heuristic time\footnote{It may initially seem that this is subsumed by Lenstra's elliptic curve method~\cite{lenstraecm}. However, the constant hidden in the big $O$ is different between the two algorithms: for~\cite{lenstraecm} it is $\sqrt{2}$, while for~\cite{mulder24} it is $1$.} $$\poly(n) \cdot \exp\left(O\left(\sqrt{m \log m}\right)\right).$$
Thus, for $n$-bit integers of the form $N = p^2 q$ with $p, q$ prime and $\log q = \widetilde{O}(n^{2/3})$, all known classical algorithms for factoring $N$ require heuristic time $$\exp\left(\widetilde{O}(n^{1/3})\right).$$
Since we consider algorithms for factorizing integers of the form $p^rq$ in the less-studied regime where $q$ is small, we note that there has been extensive classical lattice-based cryptanalysis for factoring integers of this form~\cite{DBLP:conf/crypto/BonehDH99, DBLP:conf/ctrsa/CoronFRZ16} with other special constraints. Specifically, we can factor integers of the form $p^rq^s$ in polynomial time provided that $\max(r, s) \geq \mathsf{poly}(\log p)$. We do not believe these algorithms extend to our setting; roughly, these algorithms seem to be effective in factoring $N$ when $N$ has a prime factor $p$ such that (a) $N$ is divisible by a relatively large power of $p$; and (b) $p$ is still not too small relative to $N$. In our setting where $N = p^2q$ and $q$ is very small, neither of these is the case. Regardless, we emphasize that there has been little classical cryptanalysis for factoring integers of the specific form we consider (beyond Mulder's aforementioned algorithm~\cite{mulder24}), and we leave this important direction to future work.

Finally, we remark that it is well-known that completely factoring an $n$-bit integer $N$ in the special case where $N$ is a prime power can be done in classical $\poly(n)$ time. This is because of two straightforward facts: (a) we can find integers $A, k > 1$ such that $N = A^k$ if they exist by computing $N^{1/k}$ for all possible values of $k$ (and $k_\mathrm{max} < \log N$ since $N = A^k \geq 2^k$); and (b) we can efficiently test whether $A$ is prime~\cite{Agrawal2004}.

\paragraph{Quantum Algorithms for Factoring.} 
Shor's algorithm~\cite{shor97} was the first to show that arbitrary $n$-bit integers could be factored using quantum circuits of size $\widetilde{O}(n^2)$.
However, Shor's algorithm does not benefit if $N$ has a small prime divisor or small squarefree part, like the classical algorithms by~\cite{lenstraecm, mulder24} do. 
The same holds for Regev's~\cite{Regev23} improvement on Shor's algorithm to use $\widetilde{O}(n^{3/2})$ gates.
A work by Eker{\aa} and H{\aa}stad~\cite{DBLP:conf/pqcrypto/EkeraH17}, later built upon by Chevignard et al.~\cite{chevignard_reducing_2024}, achieved a constant factor improvement in circuit costs that holds specifically for RSA integers (the product of two primes of roughly the same size), but the asymptotic scaling of the algorithms was unchanged.

To the best of our knowledge, the only polynomial-time\footnote{An alternative approach to factoring with quantum computers is to use quantum subroutines (e.g. Grover search) inside classical factoring algorithms. The benefits this yields, which can include sublinear qubit count and gains from special structure in $N$, come at the expense of \emph{superpolynomial} gate count and depth~\cite{bernstein_post-quantum_2017, bernstein_low-resource_2017, mosca_speeding_2020}.} quantum factoring circuit that benefits by more than constant factors from special structure in $N$ is the aforementioned Jacobi factoring circuit by~\cite{LPDS12}, combined with the near-linear time algorithms for computing Jacobi symbols~\cite{schonhageeuc, thull1990uni,bach1996algorithmic, Moller2008OnSA}. Putting these constructions together yields a circuit of only $\widetilde{O}(n)$ gates and space for finding a factor of $N$ when $N$ is not squarefree. Viewed in this context, one of our contributions is showing that we can further drive down the space and depth of this circuit when the squarefree part $B$ of $N = A^2B$ is much smaller than $N$.

\subsection{Sums of Phases}

\subsubsection{Basic Lemmas}

\begin{lemma}\label{lemma:phasegeometric}
    For any $x \in \RR \setminus \ZZ$ and positive integer $M$, we have:
    $$\sum_{k = 0}^{M-1} \exp\left(-2\pi i kx\right) = \frac{1 - \exp(-2\pi i xM)}{1 - \exp(-2\pi i x)}.$$
\end{lemma}
\begin{proof}
    This is just the summation formula for a geometric series. We require $x \notin \ZZ$ so that the denominator of the RHS is non-zero.
\end{proof}

\begin{lemma}\label{lemma:phasebound}
    For any $x \in \RR$, we have $|1 - \exp(2\pi i x)| = 2\cdot |\sin(\pi x)|$.
\end{lemma}
\begin{proof}
    We have:
    \begin{align*}
        |1 - \exp(2\pi i x)|^2 &= (1 - \cos(2\pi x))^2 + \sin(2\pi x)^2 \\
        &= 2 - 2\cos(2\pi x) \\
        &= 4\sin^2(\pi x).
    \end{align*}
\end{proof}

\begin{corollary}\label{cor:phasebounds}
    For any $x \in \RR$, we have $|1 - \exp(2\pi ix)| \leq 2\pi|x|$.
\end{corollary}
\begin{proof}
    This is immediate from Lemma \ref{lemma:phasebound} and the well-known inequality that $|\sin x| \leq |x|$.
\end{proof}

\begin{corollary}\label{cor:phasefarfromzero}
    For any $x \in \RR$ such that $|x| \leq 1 - \Omega(1)$, we have $|1 - \exp(2\pi i x)| = \Omega(|x|)$.
\end{corollary}
\begin{proof}
    This is immediate from Lemma \ref{lemma:phasebound} and the fact that $|\sin x| = \Omega(|x|)$ for $x \in [-\pi + \Omega(1), \pi - \Omega(1)]$.
\end{proof}
\noindent
We now combine these results in the following lemma:
\begin{lemma}\label{lemma:finalsumofphasesbound}
    For any $x \in \RR$ and positive integer $M$ such that $|xM| \leq 1-\Omega(1)$, we have $$\left|\sum_{k = 0}^{M-1} \exp(-2\pi i kx)\right| = \Theta(M).$$
\end{lemma}
\begin{proof}
    First, if $x \in \ZZ$ then each term in the summation will be 1, so the LHS will be exactly $M$. Hence we assume from now on that $x \in \RR/\ZZ$. In this case, the upper bound is straightforward: the LHS is $\leq M$ by a straightforward triangle inequality. For the lower bound, note that:
    \begin{align*}
        \left|\sum_{k = 0}^{M-1} \exp(-2\pi i kx)\right| &= \left|\frac{1 - \exp(-2\pi i xM)}{1 - \exp(-2\pi i x)}\right| \text{ (Lemma~\ref{lemma:phasegeometric})} \\
        &= \Omega\left(\frac{|xM|)}{\left|1 - \exp(-2\pi i x)\right|}\right) \text{ (Corollary~\ref{cor:phasefarfromzero})} \\
        &\geq \Omega(M) \text{ (Corollary~\ref{cor:phasebounds}).}
    \end{align*}
\end{proof}

\subsubsection{Gauss Sums}\label{sec:gausssums}

Here, we state results that essentially imply that the Jacobi symbol is appropriately ``pseudorandom'' for the purposes of our algorithm and that of~\cite{LPDS12}. We follow the lecture notes by Conrad~\cite{conrad}.

\begin{definition}[Dirichlet characters]
    For $m \in \NN$, we say that $\chi: \ZZ_m \rightarrow \CC$ is a \emph{Dirichlet character} $\bmod\text{ }m$ if the following properties all hold:
    \begin{enumerate}
        \item $\chi(a) = 0$ if and only if $\gcd(a, m) > 1$.
        \item $\chi(ab) = \chi(a)\chi(b)$ for all $a, b$.
    \end{enumerate}
\end{definition}

\begin{definition}[(Im)primitive Dirichlet characters (\cite{conrad}, Definition 3.3)]
    We say that a Dirichlet character $\chi$ mod $m$ is \emph{imprimitive} if there is a proper divisor $m'$ of $m$ and a Dirichlet character $\chi'$ mod $m'$ such that, for all $a \in \ZZ_m$ such that $\gcd(a, m) = 1$, we have $\chi(a) = \chi'(a \mod m')$.

    If $\chi$ is not imprimitive, we call it \emph{primitive}.
\end{definition}

Before continuing, we make a simple observation that the Jacobi symbol is a primitive Dirichlet character modulo any squarefree integer:

\begin{lemma}\label{lemma:jacobisprimitive}
    If $m > 1$ is odd and squarefree, then the Jacobi symbol $\chi(a) = \jac{a}{m}$ is a primitive Dirichlet character mod $m$.
\end{lemma}
\begin{proof}
    We know $\chi$ is a Dirichlet character from properties \ref{item:jacprop1} and \ref{item:jacprop2} of Theorem \ref{thm:jacobiproperties}. It remains to check that it is primitive.
    
    To this end, consider any proper divisor $m'$ of $m$ and a character $\chi'$ mod $m'$. Let $m = p_1\ldots p_r$ for distinct primes $p_1, \ldots, p_r$ (since $m$ is squarefree); since $m'$ is a proper divisor of $m$, assume without loss of generality that $p_1$ does not divide $m'$.

    Now consider $a \in \ZZ_m$ such that $a$ is a quadratic non-residue mod $p_1$ and is congruent to 1 modulo $p_2, \ldots, p_r$. Such $a$ exists by the Chinese Remainder Theorem. Then $\chi(a) = -1$. On the other hand, since $p_1$ does not divide $m'$ we have $a \bmod m' = 1 \Rightarrow \chi'(a \mod m') = \chi'(1) = 1$. (The final step is because we have $\chi'(1) = \chi'(1 \cdot 1) = \chi'(1)^2$ and $\chi'(1) \neq 0$, forcing $\chi'(1) = 1$.) Hence for this $a$, we have $\chi(a) \neq \chi'(a \bmod{m'})$. Such $a$ exists for any $m', \chi'$, so $\chi$ is indeed primitive.
\end{proof}

\begin{definition}[Gauss sums (\cite{conrad}, Definition 3.1)]
    For a Dirichlet character $\chi$ on $\ZZ_m$, we define its \emph{Gauss sum} to be $$G(\chi) = \sum_{a \in \ZZ_m} \chi(a)\exp\left(\frac{2\pi i a}{m}\right).$$
\end{definition}

\begin{theorem}[\cite{conrad}, Theorem 3.12]\label{thm:gausssum}
    For any primitive Dirichlet character $\chi$ on $\ZZ_m$, we have $|G(\chi)| = \sqrt{m}$.
\end{theorem}
\noindent
This allows us to prove the specific form of the Gauss sum bound that we will need. We refer the reader to Section~\ref{sec:lpdsoverview} for an overview of where these sums of phases come from, and reiterate the intuition here. Informally, if we want to recover $Q$ given $N = P^2Q$ as input, we will end up with a superposition of several periodic signals with period $Q$. The below lemma (with $m = Q$) examines the result of applying a QFT to this superposition, and tells us that these signals will essentially interfere like a randomly chosen collection of periodic signals.

\begin{lemma}\label{lemma:finalgausssumestimate}
    Suppose $m$ is odd and squarefree and consider any $k \in \ZZ_m$. Then we have:
    $$\left|\sum_{j \in \ZZ_m} \jac{j}{m}\exp\left(-\frac{2\pi i jk}{m}\right)\right| = \begin{cases}
        \sqrt{m}, \text{ if $\gcd(k, m) = 1$,} \\
        0, \text{ else.}
    \end{cases}.$$
\end{lemma}
\begin{proof}
    First, we address the case where $\gcd(k, m) = 1$. Let $\chi$ denote the Jacobi symbol mod $m$. In this case we make a simple change of variables:
    \begin{align*}
        \sum_{j \in \ZZ_m} \jac{j}{m}\exp\left(-\frac{2\pi i jk}{m}\right) &= \sum_{j \in \ZZ_m} \jac{-j/k}{m} \exp\left(\frac{2\pi i j}{m}\right) \\
        &= \jac{-k}{m} \cdot G(\chi),
    \end{align*}
    and now the conclusion is immediate from Theorem~\ref{thm:gausssum} and Lemma~\ref{lemma:jacobisprimitive}.
    
    It remains to address the case where $\gcd(k, m) = d > 1$. Let $k' = k/d$ and $m' = m/d$. If $k \equiv 0 \pmod{m}$, the conclusion is immediate from Corollary~\ref{cor:jacobicount}, so we may assume $m' > 1$. Since $m$ is squarefree, we have $\gcd(d, m') = 1$ and thus we can use the Chinese Remainder Theorem to associate a value $j \in \ZZ_m$ with its residues $j \bmod{m'}$ and $j \bmod{d}$. Bearing this in mind, we have:
    \begin{align*}
        \sum_{j \in \ZZ_m} \jac{j}{m} \exp\left(-\frac{2\pi i jk}{m}\right) &= \sum_{j \in \ZZ_m} \jac{j}{m'} \jac{j}{d} \exp\left(-\frac{2\pi i jk'}{m'}\right) \\
        &= \sum_{j_1 \in \ZZ_d} \sum_{j_2 \in \ZZ_{m'}} \jac{j_2}{m'} \jac{j_1}{d} \exp\left(-\frac{2\pi i j_2 k'}{m'}\right) \\
        &= \left(\sum_{j_1 \in \ZZ_d} \jac{j_1}{d}\right) \cdot \left(\sum_{j_2 \in \ZZ_{m'}} \jac{j_2}{m'} \exp\left(-\frac{2\pi i j_2 k'}{m'}\right)\right) \\
        &= 0 \cdot \left(\sum_{j_2 \in \ZZ_{m'}} \jac{j_2}{m'} \exp\left(-\frac{2\pi i j_2 k'}{m'}\right)\right) \text{ (Corollary~\ref{cor:jacobicount})} \\
        &= 0,
    \end{align*}
    as desired.
\end{proof}

\section{Factoring Squarefull Integers}\label{sec:highlevel}

In this section we prove the following theorem which refines the result by~\cite{LPDS12}. Crucially, we build on~\cite{LPDS12}, showing that it suffices for the initial superposition to extend only to $\poly(\Bmax)$, rather than $\poly(N)$ (where $\Bmax$ is an upper bound on $B$, explicitly defined in the theorem statement below). We will leverage this to factor a large class of integers with sublinear space and depth in Section~\ref{sec:jacobi}.

\begin{theorem}\label{thm:highlevel}
    Let $N, n$ be positive integers such that $2^{n-1} \leq N < 2^n$, and let $A, B$ be the unique positive integers such that $B$ is squarefree and $N = A^2B$. We will further assume that $N$ is neither squarefree nor a square i.e. $A, B > 1$. 
    Suppose there exists a quantum circuit that implements the operation $$\ket{x} \ket{0^{S}} \mapsto j_N(x) \ket{x} \ket{0^S},$$
    using $S := S(\ell, n)$ ancilla qubits with $G(\ell, n)$ gates and $D(\ell, n)$ depth, for any positive integer $x$ such that $x < 2^{\ell}$. As in Definition~\ref{def:jacobi}, $j_N(x) \in \left\{-1, 1\right\}$ is such that $j_N(x) = \jac{x}{N}$ when $\gcd(x, N) = 1$ and can be arbitrary for other $x$.

    Suppose we are also given an upper bound $\Bmax$ on $B$, and define $\ell := \floor{2\log \Bmax} + \omega(1)$. Then there is a quantum algorithm that, given as input $N$ and $\Bmax$, outputs either $B$ or a prime dividing $N$, with probability $\Omega(1)$. The quantum circuit uses $$G(\ell, n) + O(\ell \log \ell)$$gates, $D(\ell, n) + O(\ell)$ depth, and $S(\ell, n) + \ell$ qubits, and any classical pre/post-processing is polynomial-time.
\end{theorem}
\noindent
Before we prove the theorem, notice that plugging in the Jacobi symbol algorithm due to~\cite{schonhageeuc, thull1990uni,bach1996algorithmic, Moller2008OnSA} (discussed in Section~\ref{sec:cntprelims}) and simply setting $B_\mathrm{max} = N$ immediately yields the following corollary:\footnote{We note that the original paper by~\cite{LPDS12} does not appear to state a result using Sch{\"o}nhage's near-linear size Jacobi algorithm~\cite{schonhageeuc}; rather, they work with the better-known quadratic-time algorithms for computing the Jacobi symbol. Nevertheless, we credit~\cite{LPDS12} with this result since this result does essentially follow directly from their analysis (they need only set $\ell = \Theta(n)$ in Theorem~\ref{thm:highlevel}).}
\begin{corollary}[\cite{LPDS12, schonhageeuc}]\label{cor:highlevelwithschonhage}
    Let $N, A, B, n$ be as in Theorem~\ref{thm:highlevel}. Then there exists a quantum algorithm that, given as input $N$, outputs either $B$ or a prime dividing $N$, with probability $\Omega(1)$. The quantum circuit uses $\widetilde{O}(n)$ gates.
\end{corollary}

\noindent
We now turn to the proof of Theorem~\ref{thm:highlevel}. Our algorithm is detailed in Algorithm~\ref{algo:highlevel} and very closely follows Shor's period-finding algorithm~\cite{shor97}; the main difference is that we will end up with a superposition of multiple periodic signals with the same period rather than just one periodic signal. Nevertheless, we can argue using Gauss sums (see Section \ref{sec:gausssums}) that taking a QFT with this ``somewhat periodic'' signal still suffices to factor $N$.

We first address efficiency, then turn to correctness. We have the following costs:
\begin{itemize}
    \item The uniform superposition over $[1, 2^\ell]$ can be initialized in depth $1$ and gates $\ell$, using no ancilla qubits.
    \item The Jacobi symbol computation can be carried out in depth $D(\ell, n)$, gates $G(\ell, n)$, and space $S(\ell, n)$ by supposition.
    \item For the QFT mod $2^\ell$, we rely on Coppersmith's $o(1)$-approximate QFT~\cite{Cop02}, which uses $O(\ell \log \ell)$ gates, $O(\ell)$ depth, and no ancilla qubits.
\end{itemize}

As stated at the beginning of Algorithm~\ref{algo:highlevel}, let $c > 1$ be a constant parameter. We will assume throughout this section that $N = A^2B$, where $B$ is squarefree, $A, B > 1$, and any prime divisor of $A, B$ is $\geq n^c$. In our calculations, we will sometimes use big-$O$ notation to denote an arbitrary \emph{complex number} within a certain magnitude i.e. the notation $O(t)$ denotes some $z \in \CC$ such that $|z| \leq O(t)$.

\begin{algorithm}
    \SetKwInput{KwData}{Input}
    \SetKwInput{KwResult}{Output}
    \KwData{Positive integer $N = A^2B$ and a bound $\Bmax \leq N$ such that $B \in [2, \Bmax]$ is squarefree.}
    \KwResult{Either the value of $B$, or a prime divisor $p$ of $N$ (with probability $\Omega(1)$).}
    \begin{enumerate}
        \item First, we dispose of easy cases classically (in $\poly(n)$ time). Let $c > 1$ be some real constant parameter. If $N$ has any prime divisor $\leq n^c$, output that prime and terminate. We may hence assume from now on that all prime divisors of $A, B$ are $> n^{c}$, and hence also $\Bmax \geq n^c$.
        \item Set $\ell := \lfloor 2 \log \Bmax \rfloor + 1$ and $S := S(\ell, n)$.
        \item Initialize a uniform superposition $$\frac{1}{2^{\ell/2}}\sum_{x = 1}^{2^{\ell}} \ket{x} \ket{0^{S}}.$$
        \item Compute the function $j_N(x)$ in superposition, to obtain the following state: $$\frac{1}{2^{\ell/2}}\sum_{x = 1}^{2^{\ell}} j_N(x) \ket{x} \ket{0^{S}}.$$

        \item Apply a QFT mod $2^{\ell}$ to the $x$ register, and measure to obtain an integer $x^* \in [0, 2^{\ell} - 1]$.

        \item Finally, for the classical post-processing, use the continued fraction expansion of $\frac{x^*}{2^{\ell}}$ (as in Shor~\cite{shor97}; see~\cite[Chapter X]{hardy75} for details) to find positive integers $X_1$ and $X_2$ such that $X_2 \leq \Bmax$ and $\left|\frac{x^*}{2^{\ell}} - \frac{X_1}{X_2}\right|$ is minimal. Output $X_2$ and terminate. (We will show that with probability $\Omega(1)$, we will in fact have $X_2 = B$.)
    \end{enumerate}
    \caption{The Jacobi Factoring Circuit for Squarefull Integers}\label{algo:highlevel}
\end{algorithm}

It now remains to prove the correctness of Algorithm~\ref{algo:highlevel}. To do this, we first prove a preliminary technical lemma, then turn our attention to proving the theorem. The need for this technical lemma is twofold:
\begin{itemize}
    \item We would like to argue that we can safely ignore inputs $x$ in the superposition where $\gcd(x, N) > 1$, so we will have $j_N(x) \neq 0 = \jac{x}{N}$.
    \item Even when $\gcd(x, N) = 1$, we would ideally be able to say that for any $x$ in our superposition, we have $\jac{x}{N} = \jac{x}{A}^2 \jac{x}{B} = \jac{x}{B}$, and hence after measuring we end up with a superposition over values $x$ of $\jac{x}{B} \ket{x}$.
    The problem is that this is only true if $\jac{x}{A} \in \left\{-1, 1\right\}$. This is true most of the time, but there will be a small fraction of inputs $x$ (specifically, those that share common factors with $A$ but not $B$) such that $\jac{x}{B} \in \left\{-1, 1\right\}$ but $\jac{x}{A} = 0 \Rightarrow \jac{x}{N} = 0$.
\end{itemize}
The following lemma informally says that because there are not many $x$'s where either of the above issues come up, we can safely ignore these technicalities: even though our algorithm prepares the state $\ket{\psi_2}$, we can safely pretend that it in fact prepares the simpler state $\ket{\psi_1}$ (by a trace distance argument).

\begin{lemma}\label{lemma:tracedistancecoprime}
    Define $\ket{\psi_1}$ and $\ket{\psi_2}$ to be the following unnormalized states:
    \begin{align*}
        \ket{\psi_1} &= \sum_{1 \leq x \leq M} \jac{x}{B} \ket{x} \\
        \ket{\psi_2} &= \sum_{1 \leq x \leq M} j_N(x) \ket{x}.
    \end{align*}
    Then the corresponding normalized states are $O(n^{(1-c)/4} \cdot (\log n)^{1/4})$-close in trace distance. Moreover, we have $\norm{\ket{\psi_1}}_2^2 = \frac{M\varphi(B)}{B}(1+o(1)).$ (Note that these two states may not be identical, for the aforementioned reasons.)
\end{lemma}
\begin{proof}
    We will use~\cite[Lemma 2.11]{DBLP:journals/iacr/Chen24}. We first estimate $\norm{\ket{\psi_1}}_2$. We have:
    \begin{align*}
        \norm{\ket{\psi_1}}_2^2 &= \left|\left\{x \in [1, M]: \gcd(x, B) = 1\right\}\right| \\
        &= \sum_{j \in [1, B-1]: \gcd(x, B) = 1} \left|\left\{x \in [1, M]: x \equiv j \bmod{B}\right\}\right| \\
        &= \sum_{j \in [1, B-1]: \gcd(j, B) = 1} \left(\left\lfloor \frac{M-j}{B} \right\rfloor - \left\lceil \frac{1-j}{B} \right\rceil + 1\right) \text{ (Proposition~\ref{prop:countingmodsols})} \\
        &= \sum_{j \in [1, B-1]: \gcd(j, B) = 1} \left(\frac{M-j}{B} - \frac{1-j}{B} + O(1)\right) \\
        &= \frac{M-1}{B} \cdot \varphi(B) + O(\varphi(B)) \text{ (Corollary~\ref{cor:jacobicount})} \\
        &= \frac{M\varphi(B)}{B}(1 + o(1)) \text{ (noting that $M = \omega(B)$)}.
    \end{align*}
    Next, we upper bound $\norm{\ket{\psi_1} - \ket{\psi_2}}_2$. Let $p_1, \ldots, p_r$ be the distinct primes dividing $N$ but not $B$. Note that $r$ must be at most $n$, and moreover by assumption we have $p_i \geq n^c$ for all $i$. With this in mind, we have:
    \begin{align*}
        \norm{\ket{\psi_1} - \ket{\psi_2}}_2^2 &\leq O(1) \cdot \left|\left\{x \in [1, M]: \gcd(x, N) > 1\right\}\right| \\
        &\leq O(1) \cdot \sum_{i = 1}^r \left|\left\{x \in [1, M]: p_i \mid x\right\}\right| \\
        &\leq O(1) \cdot \sum_{i = 1}^r \frac{M}{p_i} \\
        &\leq O\left(\frac{M}{n^{c-1}}\right) \text{ ($r \leq n$)}.
    \end{align*}
    It then follows by~\cite[Lemma 2.11]{DBLP:journals/iacr/Chen24} that the trace distance we are concerned with is at most:
    \begin{align*}
        O\left(\sqrt{\frac{\norm{\ket{\psi_1} - \ket{\psi_2}}_2}{\norm{\ket{\psi_1}}_2}}\right) &\leq O\left(\sqrt[4]{\frac{M/n^{c-1}}{M\varphi(B)/B}}\right) \\
        &\leq O\left(n^{(1-c)/4} \cdot \left(\frac{B}{\varphi(B)}\right)^{1/4}\right) \\
        &\leq O\left(n^{(1-c)/4} \cdot (\log \log B)^{1/4}\right) \\
        &\leq O\left(n^{(1-c)/4} \cdot (\log n)^{1/4}\right),
    \end{align*}
    as desired.
\end{proof}
\noindent
We now prove Theorem~\ref{thm:highlevel}. Our proof breaks down into a few steps: we will first set up some notation and write out the state computed by the algorithm after the QFT. We then lower bound the amplitude this state places on certain values $y \in [0, 2^\ell - 1]$, and use this to complete the proof.

\paragraph{Step 1: notation and setup.} After computing the function $j_N(x)$, we have the state $$\frac{1}{2^{\ell/2}} \sum_{x = 1}^{2^\ell} j_N(x) \ket{x} \ket{0^S}.$$
We can now use Lemma \ref{lemma:tracedistancecoprime} to change this to the state $$\ket{\psi_1} = \sqrt{\frac{(1+o(1))B}{2^{\ell}\varphi(B)}} \sum_{1 \leq x \leq 2^\ell} \jac{x}{B}\ket{x},$$incurring a trace distance loss of only $O(n^{(1-c)/4}) = o(1)$. The normalization factor follows from Lemma \ref{lemma:tracedistancecoprime}. After the QFT, we obtain the state:
\begin{equation}\label{eq:postqftstateunsimplified}
    \sqrt{\frac{(1+o(1))B}{2^{2\ell}\varphi(B)}} \sum_{y = 0}^{2^{\ell} - 1} \left(\sum_{\substack{1 \leq x \leq 2^{\ell}}} \jac{x}{B} \exp\left(-\frac{2\pi i xy}{2^{\ell}}\right)\right) \ket{y}.
\end{equation}
At a high level, our analysis from this point mirrors the analysis by Shor~\cite{shor97} of his period-finding procedure; we would like to show that this state places $\Omega(1)$ weight on states $\ket{y}$ such that $y/2^{\ell}$ is close to a multiple of $1/B$ with numerator relatively prime to $B$.

\paragraph{Step 2: lower bounding the amplitude on $\ket{y}$.} In this section, we will lower bound the magnitude of the amplitude on $\ket{y}$ in Equation~\eqref{eq:postqftstateunsimplified}. Let $M = 2^{\ell}$ and $\epsilon = \frac{1}{2M}$ (this will be our target closeness bound). Note then by definition of $\ell$ (in Algorithm~\ref{algo:highlevel}) that $M > \Bmax^2$. It will be convenient for us to write $M = qB + r$, where $0 < r < B$ (we can assume $M$ is not divisible by $B$ since $B > 1$ is odd).

\begin{lemma}\label{lemma:qftlowerbound}
    Consider a fixed $y \in [0, M - 1]$ such that there exists an integer $k \in [1, B-1]$ and $\delta \in [-\epsilon, \epsilon]$ such that $\gcd(k, B) = 1$ and $\frac{y}{M} = \frac{k}{B} + \delta$. (Note in particular that this means $y \neq 0$.) Then we have:
    $$\left|\sum_{1 \leq x \leq M} \jac{x}{B} \exp\left(-\frac{2\pi i xy}{M}\right)\right| \geq \Omega(q\sqrt{B}).$$
\end{lemma}
\begin{proof}
    The high-level idea is to use the fact that $\frac{y}{M} \approx \frac{k}{B}$ to replace $\frac{y}{M}$ in the LHS with $\frac{k}{B}$. This will of course not be completely correct, but we will carefully track the errors that arise from doing this. This will allow us to obtain the desired lower bound using a Gauss sum modulo $B$ (see Lemma~\ref{lemma:finalgausssumestimate}). We proceed as follows:
    \begin{align*}
        \sum_{1 \leq x \leq M} \jac{x}{B} \exp\left(-\frac{2\pi i xy}{M}\right) &= \sum_{1 \leq j \leq B-1} \sum_{\substack{1 \leq x \leq M\\x \equiv j \bmod{B}}} \jac{j}{B} \exp\left(-\frac{2\pi i xy}{M}\right) \\
        &= \sum_{1 \leq j \leq B-1} \sum_{0 \leq l \leq \floor*{\frac{M-j}{B}}} \jac{j}{B} \exp\left(-\frac{2\pi i (lB+j)y}{M}\right) \text{ (writing $x = lB+j$)} \\
        &= \sum_{j \in [1, B-1]} \left[\jac{j}{B} \exp\left(-\frac{2\pi i jy}{M}\right) \cdot \sum_{0 \leq l \leq \floor*{\frac{M-j}{B}}} \exp\left(-\frac{2\pi i lBy}{M}\right)\right].\stepcounter{equation}\tag{\theequation}\label{eqn:giantfourierexpression}
    \end{align*}
    We now analyze the inner sum. Note that $M-j = qB + r-j \Rightarrow \floor*{\frac{M-j}{B}} \in \left\{q-1, q\right\}$. We hence have:
    \begin{align*}
        \sum_{l = 0}^{\floor*{\frac{M-j}{B}}} \exp\left(-\frac{2\pi i lBy}{M}\right) &= \sum_{l = 0}^{q-1} \exp\left(-\frac{2\pi i lBy}{M}\right) + O(1) \\
        &= \sum_{l = 0}^{q-1} \exp\left(-2\pi i lB \left(\frac{k}{B} + \delta\right)\right) + O(1) \\
        &= \sum_{l = 0}^{q-1} \exp\left(-2\pi i lB\delta\right) + O(1)\\
        &= R + O(1),
    \end{align*}
    \noindent
    where we define $R := \sum_{l = 0}^{q-1} \exp(-2\pi i lB \delta)$. Since $|qB\delta| \leq qB\epsilon \leq 1/2$, we have by Lemma~\ref{lemma:finalsumofphasesbound} that $|R| = \Theta(q)$. Bearing this in mind, we plug this into and continue from Equation~\eqref{eqn:giantfourierexpression} as follows:
    \begin{align*}
        &\sum_{j \in [1, B-1]} \left[\jac{j}{B} \exp\left(-\frac{2\pi i jy}{M}\right) \cdot \sum_{0 \leq l \leq \floor*{\frac{M-j}{B}}} \exp\left(-\frac{2\pi i lBy}{M}\right)\right] \\
        = &\sum_{j \in [1, B-1]} \left[\jac{j}{B}\exp\left(-\frac{2\pi i jy}{M}\right) \cdot \left(R + O(1)\right)\right] \\
        =& R \cdot \left[\sum_{j \in [1, B-1]} \jac{j}{B}\exp\left(-\frac{2\pi i jy}{M}\right)\right] + O(\varphi(B)) \\
        =& R \cdot \left[\sum_{j \in [1, B-1]} \jac{j}{B}\exp\left(-2\pi i j\left(\frac{k}{B} + \delta\right)\right)\right] + O(\varphi(B)) \\
        =& R \cdot \left[\sum_{j \in [1, B-1]} \jac{j}{B} \exp\left(-\frac{2\pi i jk}{B}\right)(1 + O(B\epsilon))\right] + O(\varphi(B)) \text{ (Corollary \ref{cor:phasebounds})} \\
        =& R \cdot \left[\left(\sum_{j \in [1, B-1]} \jac{j}{B}\exp\left(-\frac{2\pi i jk}{B}\right)\right) + O(B\varphi(B)\epsilon)\right] + O(\varphi(B)) \\
        =& R \cdot \left[\sqrt{B} + O(1)\right] + O(\varphi(B)) \text{ (Lemma \ref{lemma:finalgausssumestimate}).}
    \end{align*}
    \noindent
    The final step follows from Gauss sums; we state their key properties (including Lemma~\ref{lemma:finalgausssumestimate}) in Section~\ref{sec:gausssums}. We also use the fact that $B\varphi(B)\epsilon \leq B^2/M \leq 1$. Finally, we lower bound the magnitude of this amplitude as follows:
    \begin{align*}
        &\left|R \cdot \left[\sqrt{B} + O(1)\right] + O(\varphi(B))\right| \\
        \geq &\left|R \cdot \sqrt{B}\right| - O\left(\left|R\right|\right) - O(\varphi(B)) \\
        \geq &\Omega(q\sqrt{B}) - O(q) - O(\varphi(B))\text{ (since $|R| = \Theta(q)$)} \\
        \geq &\Omega(q\sqrt{B}),
    \end{align*}
    as desired. To justify the final step, note that $q\sqrt{B} \geq q \sqrt{n^c} \gg q \geq B \geq \varphi(B)$.
\end{proof}

\begin{proof}[Proof of Theorem \ref{thm:highlevel}]
    Call $y \in [0, M-1]$ \emph{successful} if there exists an integer $k \in [1, B-1]$ and real $\delta \in [-\epsilon, \epsilon]$ such that $\gcd(k, B) = 1$ and $\frac{y}{M} = \frac{k}{B} + \delta$. By Lemma~\ref{lemma:qftlowerbound} and Equation~\eqref{eq:postqftstateunsimplified}, for any successful $y$, the probability that we measure and get the classical outcome $x^* = y$ is at least
    $$\Omega\left(\frac{B}{M^2\varphi(B)} \cdot (q\sqrt{B})^2\right) = \Omega\left(\frac{q^2B^2}{M^2\varphi(B)}\right) \geq \Omega\left(\frac{1}{\varphi(B)}\right).$$
    Next, we argue that there are at least $\varphi(B)$ successful values of $y$. Indeed, for any $k \in [1, B-1]$ such that $\gcd(k, B) = 1$, consider the interval $\left[\frac{k}{B} - \epsilon, \frac{k}{B} + \epsilon\right] \subset (0, 1)$. It has width $2\epsilon = \frac{1}{M}$, so there must be at least one multiple of $1/M$ in this interval i.e. there exists an integer $y_k$ such that $|\frac{y_k}{M} - \frac{k}{B}| \leq \epsilon$ (which in turn implies $y_k/M \in (0, 1) \Rightarrow y_k \in [1, M-1]$); in other words, $y_k$ is successful. Moreover, we claim that $y_k \neq y_{k'}$ for any $k \neq k'$. If this were not true, then the triangle inequality would force $\left|\frac{k}{B} - \frac{k'}{B}\right| \leq 2\epsilon \Rightarrow \frac{1}{B} \leq 2\epsilon$, which is false. Since there are $\varphi(B)$ many such values of $k$, there are at least $\varphi(B)$ distinct successful values of $y$, each of which we obtain with probability $\geq \Omega(\frac{1}{\varphi(B)})$. It follows that with $\Omega(1)$ probability, we will obtain such a $y$, as claimed.
    
    Now to finish, we have some integer $y = x^*$ such that $\left|\frac{y}{M} - \frac{k}{B}\right| \leq \frac{1}{2M} < \frac{1}{2\Bmax^2}$. It follows that $k/B$ is the closest fraction to $\frac{y}{M}$ with denominator at most $\Bmax$. Hence our algorithm will obtain $X_1 = k$ and $X_2 = B$ in the final step, and output the denominator $B$. This completes the proof of the theorem.
\end{proof}

\begin{remark}
	Algorithm~\ref{algo:highlevel}, and its associated Theorem~\ref{thm:highlevel}, receive as input a bound $\Bmax$ on the size of $B$, the squarefree part of the input integer $N$.
	Here we note that if $\Bmax$ is not known, the algorithm can still be used to find $B$ (or a prime dividing $N$) with high probability, and with roughly the same quantum circuit sizes, as follows.
	Via Lemma~\ref{lemma:highlevelboosting}, for any $\Bmax > B$ the probability of success of the algorithm can be boosted to $1-\epsilon$ using $O(\log 1/\epsilon)$ calls to Algorithm~\ref{algo:highlevel} and a small amount of classical computation.
	Starting with some $\Bmax = O(1)$, this larger algorithm can then be iterated, doubling $\log \Bmax$ every iteration until a value that divides $N$ is found and the algorithm halts.
	The number of iterations is expected to be $\log \log B$ and with high probability the algorithm will halt with $\log \Bmax < 2 \log B$ on the last iteration.
	This implies that the complexity of the algorithm will be only worse by a constant factor if $\Bmax$ is not supplied.
\end{remark}

\begin{remark}
    In Algorithm~\ref{algo:highlevel}, we take a superposition over all $x < O(\Bmax^2)$ in order to recover the period of a periodic function with period $\leq \Bmax$. This is in direct analogy with Shor's original period-finding subroutine~\cite{shor97}: to factor an integer $N$, Shor considers a periodic function with period $\leq N$ and takes a superposition over all inputs $x < O(N^2)$ to recover this period.

    Works subsequent to Shor's original paper~\cite{seifert2001using, DBLP:conf/pqcrypto/EkeraH17} show that it can suffice to use a superposition only over all $x < O(N^{1+\epsilon})$ for any $\epsilon > 0$. 
    The circuit must then be run independently $O(1/\epsilon)$ times; the period is subsequently recovered using a more sophisticated classical post-processing procedure. Analogously, we believe it is likely possible to modify Algorithm~\ref{algo:highlevel} to only take the superposition up to $O(\Bmax^{1+\epsilon})$, and run the resulting circuit $O(1/\epsilon)$ times and classically post-process the results as in~\cite{seifert2001using}. This would enable constant-factor improvements to the space and depth, which would be important when instantiating this circuit in practice as a proof of quantumness (see Corollary~\ref{corr:factoringpoq}).
\end{remark}

\begin{remark}[On the use of Gauss sums]\label{remark:gausssums}
    In this analysis, we made use of the ``strong half'' of Lemma~\ref{lemma:finalgausssumestimate}, namely its conclusion in the case where $\gcd(k, m) = 1$.
    On the other hand, the original analysis by~\cite{LPDS12} only makes use of the ``weak half'' i.e. the case where $\gcd(k, m) > 1$.
    The reason we use Gauss sums in all their power is because of our adaptation of Shor's analysis~\cite{shor97} to the setting where the initial superposition ranges only up to $\poly(Q)$. As in Shor's analysis, we lower bounded the amplitude of the final (post-QFT) state on each ``useful'' value and this requires a strong Gauss sum bound.
    
    If we wanted to get around the need for the strong Gauss sum bound, we could instead adapt the tighter analysis of Regev's factoring algorithm~\cite{Regev23} based on the Poisson summation formula. This would only require the ``weak'' Gauss sum result (as in~\cite{LPDS12}), but would likely require us to start with a discrete Gaussian superposition instead of a uniform superposition.
    Given the concrete overheads involved in preparing discrete Gaussian superpositions~\cite{grover2002creating, regev09} and our interest in the potential practicality of our algorithms, we chose to present the simpler algorithm with just a uniform superposition.

    We thank Oded Regev for pointing out to us that the use of ``strong'' Gauss sum bounds in our analysis could be circumvented.
\end{remark}

\section{Algorithm for Computing Jacobi Symbols}\label{sec:jacobi}

In this section we present one of our core technical contributions: an algorithm to compute the Jacobi symbol of $x$ mod $N$, where $N$ is classical and $x$ could be in superposition. We remark that our algorithm is also readily adaptable to computing the $\gcd$ of $x$ and $N$, much like other algorithms for computing the Jacobi symbol~\cite{schonhageeuc, thull1990uni,bach1996algorithmic, Moller2008OnSA}. When $N < 2^n$ and $x < 2^m$, our construction requires circuit-size $\widetilde{O}(n)$ and space $\widetilde{O}(m)$, which we can exploit due to our analysis in Section~\ref{sec:highlevel} which allows us to restrict $m \ll n$. 
In contrast, the 2012 result of Li, Peng, Du and Suter \cite{LPDS12}, together with near-linear time algorithms due to~\cite{schonhageeuc, thull1990uni,bach1996algorithmic, Moller2008OnSA} for computing Jacobi symbols, uses gates \emph{and space} $\widetilde{O}(n)$ (see Corollary~\ref{cor:highlevelwithschonhage} for a formal statement of this result by~\cite{LPDS12}). Our improvements are thus along the axes of space and, as we will see in Section~\ref{sec:jacobicorollaries}, depth.

We begin with an abstract algorithm (formalized in Theorem~\ref{thm:jacobimain}) that makes black-box use of circuits for multiplying and for computing the Jacobi symbol between equally-sized inputs. We then instantiate the circuits using explicit constructions for these subroutines~\cite{schonhageeuc, thull1990uni,bach1996algorithmic, Moller2008OnSA, nie_quantum_2023} in Section~\ref{sec:jacobicorollaries}.

\subsection{Abstract Construction}

Let us first summarize the main idea of our construction; we refer the reader to the technical overview in Section~\ref{sec:techoverview} for further discussion of our high-level approach. 

The standard algorithms for computing the Jacobi symbol are the extended Euclidean algorithm and the binary GCD; out of the box, neither one achieves the efficiency we desire, in particular when one input is much smaller than the other. 
Indeed, existing circuits for both algorithms require space proportional to the length of the larger input.
Nevertheless, we can draw inspiration from an observation about the extended Euclidean algorithm: after just one iteration, the larger input is reduced to roughly the same size as the smaller one, and the entire rest of the computation has cost that scales only with the length of the smaller input.
Our task thus boils down to ``just'' computing $N\bmod{x}$ in sublinear space.
The naive idea would be to use standard long division or a variant thereof.
However, while long division only needs to look at $O(\log(x))$ bits of $N$ at a time, we have to keep track of the intermediate values to make the computation reversible. Sadly, keeping track of these excess values appears to require linear space. 

Our solution, which is inspired by the binary GCD algorithm~\cite{bach1996algorithmic} and strongly resembles Montgomery reduction~\cite{Montgomery1985}, is to perform a ``flipped'' variant of long division.
In the Euclidean algorithm, the computation of $N \bmod x$ can be thought of as a way to find a multiple $\kappa x$ such that $N-\kappa x < x$ ---
 in particular, all but the lowest $m$ bits of $N-\kappa x$ are zero.
Long division achieves this by starting from the most significant bit (MSB) and iterating towards the least significant bit (LSB), zeroing out bits along the way.
We flip this idea on its head: we compute a multiple $kx$ such that all but the \emph{highest-order} $m$ bits of $N-kx$ are zero---by starting from the LSB and iterating towards the MSB. 
From here, there are three key ideas:
\begin{enumerate}
    \item The value $kx$ can be built up bit by bit via a loop in which each iteration is entirely reversible.
    We can easily uncompute intermediate states by performing a simple comparison that depends only on the MSBs of our current state.
    \item Furthermore, once a lower-order bit of $kx$ matches that of $N$, this bit will not depend on anything quantum (since $N$ is classical).
    This qubit is thus in a classically-known state, and unentangled from all other qubits in the computation.
    We can use this fact to recycle lower-order qubits as we set them.
    Ultimately, we only ever store the leading $O(m)$ bits of the partial value $y = kx$ that is being computed. 
    In Algorithm~\ref{algo:match_bits_of_N}, we will denote the register holding this sliding window of $O(m)$ bits as $z$.
    \item The above two ideas on their own already suffice to obtain a circuit that computes the Jacobi symbol in sublinear space $O(m)$. However, its gate count will be $O(nm)$. This is because, when setting one bit of $k$ at a time, we will ultimately have to carry out ``schoolbook'' arithmetic comprising additions.

    Instead, we can set $k$ in \emph{batches} of $m$ bits each. This allows us to benefit from fast integer multiplication algorithms~\cite{schonhage1971fast, Harvey21, kahanamokumeyer2024fast}, and will ultimately drop the gate count to $\widetilde{O}(n)$, which is near-linear.
\end{enumerate}
\noindent
At the end, our quantum computer will hold the value $(N-kx)/2^{n-m}$, where $kx$ is a multiple such that both of the following are true (where we let $m \geq \lceil \log x \rceil$): (a) $N - kx$ is divisible by $2^{n-m}$; and (b) $kx < 2^n$. We can now proceed from here in one of two ways:
\begin{itemize}
    \item This actually suffices to complete the first step of extended Euclidean and compute $N \bmod{x}$, since we have:
    \begin{align*}
        N &\equiv \frac{N - kx}{2^{n-m}} \cdot 2^{n-m} \pmod{x},
    \end{align*}
    and $2^{n-m} \bmod{x}$ is very efficiently computable. This is asymptotically a satisfactory solution, but adds unnecessary indirection and concrete efficiency to our overall algorithm, as we will see next.

    We thank Daniel J. Bernstein for pointing out this variant of our algorithm to us.

    \item Rather than taking the extra step as above to compute $N \bmod{x}$, we could just start from $(N - kx)/2^{n-m}$ and directly carry out the following chain of transformations to compute $\jac{x}{N}$, using the properties of the Jacobi symbol stated in Theorem~\ref{thm:jacobiproperties}:
    $$\jac{x}{N} \rightarrow \jac{N}{x} \rightarrow \jac{N-kx}{x} \rightarrow \jac{(N - kx)/2^{n-m}}{x}.$$
    This is the approach that we will take, as detailed in Algorithm~\ref{algo:jacobi}.
\end{itemize}

\noindent
We first describe how to utilize the value $kx$ to compute the Jacobi symbol (Algorithm \ref{algo:jacobi}), then present our procedure for obtaining the value $kx$ (Algorithm \ref{algo:match_bits_of_N}), and then finally prove our claimed performance guarantees.

\begin{algorithm}
    \caption{Reversible algorithm for computing Jacobi symbols}\label{algo:jacobi}
    \KwData{Efficiency parameters $m, n$ such that $m | n$, and positive integers $N < 2^n$ and $x < 2^m$}
    \KwResult{The Jacobi symbol $\left(\frac{x}{N}\right)$}
    \begin{enumerate}
        \item Compute the integers $x'$ and $t$ such that $x' = x/2^t$ is an odd integer.\label{algostep:jacobi:makeodd}
        \item Set the register $\mathsf{out}$, which will ultimately store the Jacobi symbol $\left(\frac{x}{N}\right)$, as follows: \label{algostep:jacobi:setout}
        \begin{equation}\label{eqn:init_out_alg_jacobi}
            \mathsf{out} = \left( \left( -1\right)^{\frac{N^2 - 1}{8} } \right)^t  
            \cdot \left(-1\right)^{\frac{(x'-1)(N-1)}{4}} 
            \cdot \left( \left( -1\right)^{\frac{x'^2 - 1}{8}} \right)^{n-m}.
        \end{equation}
        \item Use Algorithm~\ref{algo:match_bits_of_N} on inputs $N$ and $x'$ to compute some integer $z$, then set $s = \left\lfloor \frac{N}{2^{n-m}} \right\rfloor - z$. (We will show in Lemma~\ref{lemma:algojacobicorrect} that $s = (N - kx')/2^{n-m}$ for some integer $k$. Note that $s$ could be negative.) \label{algostep:jacobi:matchbits}
        
        \item Compute $\mathsf{out} \leftarrow \mathsf{out} \cdot \left(\frac{s}{x'}\right)$, where the Jacobi symbol  $\left(\frac{s}{x'}\right)$ is computed via the algorithms of \cite{schonhageeuc,thull1990uni,bach1996algorithmic, Moller2008OnSA}, made reversible via standard techniques~\cite{bennett_logical_1973, bennett_timespace_1989, levine_note_1990}.\label{algostep:jacobi:msize}
        \item Uncompute $s$, $z$, $x'$, and $t$ by running steps~\ref{algostep:jacobi:matchbits} and~\ref{algostep:jacobi:makeodd} in reverse. \label{algostep:jacobi:uncompute}
        \item Return $\mathsf{out}$.
    \end{enumerate}
\end{algorithm}

\begin{algorithm}
    \caption{Reversible subroutine for finding a value $kx$, whose $n-m$ lowest-order bits equal the corresponding bits of $N$}\label{algo:match_bits_of_N}
    \KwData{Efficiency parameters $m, n$ such that $m | n$, and positive integers $N < 2^n$ and $x < 2^m$ with $x$ odd}
    \KwResult{The highest-order $m$ bits of $y = kx$ for some $k$, such that $kx < 2^n$ and $N-kx$ is divisible by $2^{n-m}$.}
    \begin{enumerate}
        \item Set a $2m$ bit register $z = 0$. The low-order half of this register will ultimately store the leading $m$ bits of $y$. (All other, lower-order bits of $y$ match the corresponding bits of $N$, and thus do not need to be stored explicitly.)
        \item Precompute the following values:
        \begin{enumerate}\label{algostep:match_bits_of_N:precomputation}
            \item $x_\text{minv}$, the inverse of $x \bmod 2^m$
            \item $x_\text{inv} = \frac{1}{x}$ with $2m$ bits of precision
        \end{enumerate}
        \item Repeat the following for $j \in \{0,1,2,... \frac{n-2m}{m} \}$
        \begin{enumerate}\label{algostep:match_bits_of_N:loop}
            \item Compute an $m$-bit register $\mathsf{ctrl} = \left[ x_\text{minv} \cdot (N_j - z) \right] \bmod 2^m$, where $N_j = \lfloor N/2^{jm} \rfloor \bmod 2^m$.
            \item Compute $z \leftarrow z + \mathsf{ctrl} \cdot x$. Now, $z \bmod 2^m = N_j$.
            \item Using $z$ and $x_\mathrm{inv}$, uncompute $\mathsf{ctrl} \leftarrow \mathsf{ctrl} \oplus \lfloor \frac{z}{x}\rfloor$ via \cite[Lemma A.2]{rv24}.
            \item Zero the $m$ lowest-order bits of $z$ using $N_j$, then swap them with the highest-order $m$ bits: $z \leftarrow \lfloor \frac{z}{2^m} \rfloor $.
        \end{enumerate}
        \item Uncompute the values from Step~\ref{algostep:match_bits_of_N:precomputation}.\label{algostep:match_bits_of_N:unprecomputation}
        \item Return $z$.
    \end{enumerate}
\end{algorithm}

\begin{theorem}\label{thm:jacobimain}
Suppose there exists a quantum multiplication circuit on $t$-bit inputs with gates $G_\mathrm{mult}(t)$, space $S_\mathrm{mult}(t)$, and depth $D_\mathrm{mult}(t)$. Also, suppose there exists a quantum circuit for computing the Jacobi symbol between two $t$-bit inputs with gates $G_\mathrm{Jac}(t)$, space $S_\mathrm{Jac}(t)$, and depth $D_\mathrm{Jac}(t)$.

Let $N < 2^n$ be a classically known odd integer. Then, there exists a quantum circuit implementing the unitary
\begin{equation}
\ket{x} \ket{0}^{\otimes 2} \mapsto \ket{x} \ket{\jac{x}{N}},
\label{eq:jac-unitary}
\end{equation}
\sloppy acting on $m$-qubit quantum inputs $x \in [0, 2^m-1]$ that runs in gates $O\left(\frac{n}{m} \cdot G_\mathrm{mult}(m) + G_\mathrm{Jac}(m) + m \log m\right)$, space $O\left(\max(S_\mathrm{mult}(m), S_\mathrm{Jac}(m))\right)$, and depth $O\left((\frac{n}{m} + \log m) \cdot D_\mathrm{mult}(m) + D_\mathrm{Jac}(m) + \log^2 m\right)$.
\end{theorem}
\begin{remark}
    We state our result in terms of explicitly writing down the Jacobi symbol for the purposes of clarity and to emphasize that this also works equally well in the classical reversible setting.
    However, as detailed in Section~\ref{sec:highlevel}, our algorithms will ultimately apply the Jacobi symbol in the \emph{phase}, thus computing the unitary $\ket{x} \mapsto j_N(x) \ket{x}$ for some function $j_N(x)$ that is equal to the Jacobi symbol whenever $\gcd(x, N) = 1$.
    This can easily be achieved either by using the above algorithm out of the box and applying a phase kickback~\cite{cemm}, or by going through Algorithms~\ref{algo:jacobi} and~\ref{algo:match_bits_of_N} and modifying them to write the output in the phase rather than in a single-qubit register.
\end{remark}

\paragraph{Correctness.} We first show the correctness of Algorithm~\ref{algo:match_bits_of_N}. For each value of the iteration index $j = 0, 1, \ldots, \frac{n-2m}{m}$ for the loop in step~\ref{algostep:match_bits_of_N:loop}, we define the following variables:
\begin{itemize}
    \item $z_j$: the value stored in register $z$ at the beginning of iteration $j$ of the loop;
    \item $\mathsf{ctrl}_j$: the value of $\mathsf{ctrl}$ computed in iteration $j$;
    \item $z'_j$: the intermediate value of $z$ in iteration $j$ i.e. $z_j + \mathsf{ctrl}_j \cdot x$; and
    \item $y_j$: this is defined as $z_j \cdot 2^{jm} + (N \bmod 2^{jm})$. This is the multiple of $x$ that we are tracking implicitly throughout the algorithm; we use $y_j$ to represent the value of this multiple at the beginning of iteration $j$.
\end{itemize}
Also note the definition of $N_j$ in step~\ref{algostep:match_bits_of_N:loop} of Algorithm~\ref{algo:match_bits_of_N}, and let $z_{(n-m)/m}$ denote the value stored in register $z$ at the end of the algorithm i.e. the final output. As we will see, correctness of Algorithm~\ref{algo:match_bits_of_N} boils down to the following lemma:

\begin{lemma}\label{lemma:match_bits_induction}
    For all $j = 0, 1, \ldots, (n-m)/m$, all of the following hold:
    \begin{enumerate}
        \item\label{item:congmodpow2} $y_j \equiv N \pmod{2^{jm}}$;
        \item\label{item:bounded} $0 \leq y_j < 2^{jm} \cdot x$; and
        \item\label{item:divbyx} $y_j$ is divisible by $x$.
    \end{enumerate}
\end{lemma}
\begin{proof}
    Item~\ref{item:congmodpow2} is straightforward. To establish the other two items, we proceed by induction on $j$. For the base case where $j = 0$, we have $z_0 = 0 \Rightarrow y_0 = 0$. All three conditions are now evident. Now, for the inductive step, assume that we have shown the result for some $j \geq 0$ and wish to show the result for $j+1$. 
    We will examine the execution of iteration $j$ of step~\ref{algostep:match_bits_of_N:loop} of Algorithm~\ref{algo:match_bits_of_N}  to complete the induction. Firstly, by definition we have the following:
    \begin{align*}
        z'_j &= z_j + \mathsf{ctrl}_j \cdot x \\
        &\equiv z_j + x_{\mathrm{minv}} \cdot (N_j - z_j) \cdot x \pmod{2^m} 
        \\
        &\equiv N_j \pmod{2^m} \\
        \Rightarrow 2^{jm} z'_j &\equiv 2^{jm} N_j \pmod{2^{(j+1)m}} \\
        \Rightarrow 2^{jm} z'_j + \left(N \bmod 2^{jm}\right) &\equiv 2^{jm} N_j + \left(N \bmod{2^{jm}}\right) \pmod{2^{(j+1)m}} \\
        &\equiv N \pmod{2^{(j+1)m}}.\stepcounter{equation}\tag{\theequation}\label{eqn:congruencepartial}
    \end{align*}
    This implies that the $m$ lowest-order bits of $z'_j$ match $N_j$ and thus we can indeed use the bits of $N_j$ to zero out the $m$ lowest-order bits of $z'_j$. Hence we have $$2^{jm} z'_j + (N \bmod 2^{jm}) = 2^{(j+1)m} z_{j+1} + (N \bmod 2^{(j+1)m}),$$
    Now, we also have that:
    \begin{align*}
        y_{j+1} &= 2^{(j+1)m} z_{j+1} + (N \bmod 2^{(j+1)m}) \\
        &= 2^{jm} z'_j + (N \bmod 2^{jm}) \\
        &= 2^{jm} (z_j + \mathsf{ctrl}_j \cdot x) + (N \bmod 2^{jm}) \\
        &= y_j + 2^{jm} \cdot \mathsf{ctrl}_j \cdot x.\stepcounter{equation}\tag{\theequation}\label{eqn:yrecurrence}
    \end{align*}
    Since $y_j$ is divisible by $x$ by the induction hypothesis, this immediately implies condition~\ref{item:divbyx}. Finally, we can obtain condition~\ref{item:bounded} since:
    \begin{align*}
        y_{j+1} &= y_j + 2^{jm} \cdot \mathsf{ctrl}_j \cdot x \\
        &< 2^{jm} \cdot x + 2^{jm} \cdot \mathsf{ctrl}_j \cdot x \text{ (induction hypothesis)} \\
        &\leq 2^{jm} \cdot x + 2^{jm} \cdot (2^m - 1) \cdot x \text{ (since $\mathsf{ctrl}_j$ is reduced mod $2^m$)} \\
        &= 2^{(j+1)m}x.
    \end{align*}
\end{proof}
\noindent
We complete our proof of correctness for Algorithm~\ref{algo:match_bits_of_N} with the following claim:
\begin{proposition}\label{prop:zbounds}
    For all $j = 0, 1, \ldots, (n-m)/m$, we have $0 \leq z_j < x$. For $j = 0, 1, \ldots, (n-2m)/m$, we have $0 \leq z'_j < 2^{2m}$. Moreover, we have $\mathsf{ctrl}_j = \left\lfloor \frac{z'_j}{x} \right\rfloor$. (This establishes that $2m$ qubits are sufficient to hold the $z$ register, and that the uncomputation of $\mathsf{ctrl}_j$ proceeds correctly.)
\end{proposition}
\begin{proof}
    Since $y_j = z_j \cdot 2^{jm} + (N \bmod 2^{jm})$, we have:
    \begin{align*}
        z_j &= \left\lfloor \frac{y_j}{2^{jm}} \right\rfloor \\
        &\leq \left\lfloor \frac{(2^{jm}-1)x}{2^{jm}} \right\rfloor \text{ (Lemma~\ref{lemma:match_bits_induction})} \\
        &\in [0, x-1].
    \end{align*}
    Since $z'_j = z_j + \mathsf{ctrl}_j \cdot x$, it follows that $\mathsf{ctrl}_j = \left\lfloor \frac{z'_j}{x} \right\rfloor$. Finally, we have:
    \begin{align*}
        z'_j &= z_j + \mathsf{ctrl}_j \cdot x \\
        &< (1 + \mathsf{ctrl}_j) \cdot x \\
        &\leq 2^m \cdot x \text{ ($\mathsf{ctrl}_j$ is reduced mod $2^m$)} \\
        &< 2^{2m},
    \end{align*}
    as desired.
\end{proof}
\noindent
Finally, we show the correctness of Algorithm~\ref{algo:jacobi}:
\begin{lemma}\label{lemma:algojacobicorrect}
    Algorithm~\ref{algo:jacobi} correctly computes the Jacobi symbol $\jac{x}{N}$. Moreover, we have $|s| < 2^m$ (thus the step of computing $\jac{s}{x'}$ only needs to work with $m$-bit inputs).
\end{lemma}
\begin{proof}
    We retain all notation introduced in Algorithm~\ref{algo:jacobi}. We first show that $|s| < 2^m$ and that there exists an integer $k$ such that $N - kx' = 2^{n-m} \cdot s$. To this end, recall that the output of Algorithm~\ref{algo:match_bits_of_N} is exactly $z = z_{(n-m)/m}$, where $z_{(n-m)/m} \cdot 2^{n-m} + (N \bmod 2^{n-m}) = y_{(n-m)/m}$ is equal to $kx'$ for some integer $k$ by Lemma~\ref{lemma:match_bits_induction}. Then note firstly that:
    \begin{align*}
        |s| &\leq \max\left(\left\lfloor \frac{N}{2^{n-m}} \right\rfloor, z_{(n-m)/m}\right) \\
        &< 2^m,
    \end{align*}
    since $N < 2^n$ and $z_{(n-m)/m} < x \leq 2^m$ by Proposition~\ref{prop:zbounds}. Secondly, we have:
    \begin{align*}
        2^{n-m} \cdot s &= 2^{n-m} \cdot \left(\left\lfloor \frac{N}{2^{n-m}} \right\rfloor - z_{(n-m)/m}\right) \cdot \\
        &= 2^{n-m} \cdot \left\lfloor \frac{N}{2^{n-m}} \right\rfloor - 2^{n-m} \cdot z_{(n-m)/m} \\
        &= N - \left(N \bmod 2^{n-m}\right) - 2^{n-m} \cdot z_{(n-m)/m} \\
        &= N - y_{(n-m)/m} \\
        &= N - kx'.
    \end{align*}
    The conclusion will now follow from the standard properties of the Jacobi symbol stated in Theorem~\ref{thm:jacobiproperties}:
    \begin{align*}
        \left( \frac{x}{N} \right) &= \left( \frac{2^tx'}{N} \right) \\
        &= \left( \left( -1\right)^{\frac{N^2 - 1}{8} } \right)^t  \cdot \left( \frac{x'}{N} \right) \text{ (Theorem~\ref{thm:jacobiproperties}, properties~\ref{item:jacprop2} and~\ref{item:jacprop6})} \\
        &= \left( \left( -1\right)^{\frac{N^2 - 1}{8} } \right)^t  \cdot  \left(-1\right)^{\frac{(x'-1)(N-1)}{4}} \cdot \left( \frac{N}{x'} \right) \text{ (Theorem~\ref{thm:jacobiproperties}, property~\ref{item:jacprop7})} \\
        &= \left( \left( -1\right)^{\frac{N^2 - 1}{8} } \right)^t  \cdot  \left(-1\right)^{\frac{(x'-1)(N-1)}{4}} \cdot \left( \frac{N-kx'}{x'} \right) \text{ (Theorem~\ref{thm:jacobiproperties}, property~\ref{item:jacprop4})} \\
        &= \left( \left( -1\right)^{\frac{N^2 - 1}{8} } \right)^t  \cdot  \left(-1\right)^{\frac{(x'-1)(N-1)}{4}} \cdot \left( \frac{2^{n-m} \cdot s}{x'} \right) \\
        &= \left( \left( -1\right)^{\frac{N^2 - 1}{8} } \right)^t  \cdot  \left(-1\right)^{\frac{(x'-1)(N-1)}{4}} \cdot 
        \left( \left( -1\right)^{\frac{x'^2 - 1}{8} } \right)^{n-m} \cdot
        \left( \frac{s}{x'} \right), \text{ (Theorem~\ref{thm:jacobiproperties}, properties~\ref{item:jacprop2} and~\ref{item:jacprop6})}
    \end{align*}
    which implies the conclusion.
\end{proof}

\paragraph{Efficiency.} We now turn our attention to establishing the desired efficiency guarantees:
\begin{lemma}
    Algorithm~\ref{algo:jacobi} achieves the efficiency guarantees claimed in Theorem~\ref{thm:jacobimain}.
\end{lemma}
\begin{proof}
    We proceed by showing that each step of Algorithm~\ref{algo:jacobi} can be implemented reversibly with the stated complexities.
    We note that $S_\mathrm{Jac}(m)$ must be at least linear in $m$ because the Jacobi symbol depends on the entire input; therefore our space complexity is lower bounded by $\Omega(m)$.

    Step~\ref{algostep:jacobi:makeodd} computes the number of trailing zeros $t$ of a length-$m$ value $x$, and then computes $x'$, which is $x$ shifted to the right by $t$ bits.
    This can be performed reversibly in $O(m \log m)$ gates, $O(m)$ space, and $O(\log^2 m)$ depth, as follows. 
    We first compute $t$ by using a tree of $O(m)$ Toffoli gates to compute the unary representations of $\lfloor t/2^i \rfloor$ for $i$ from $1$ up to $\lceil \log m \rceil$, and then a tree of controlled-NOT gates to compute the parity of each unary value, which is equal to bit $i$ of $t$.
    The value $x'$ can then be computed by applying the map $\ket{a} \to \ket{\lfloor a/2^{i t_i} \rfloor}$ repeatedly for each bit $t_i$ of $t$, beginning with $a=x$.
    In turn, this map can be implemented with $m$ ancilla bits by applying the out-of-place controlled bit-shift map $\ket{a}\ket{0} \to \ket{a}\ket{\lfloor a/2^{i t_i} \rfloor}$, followed by $\ket{a}\ket{b} \to \ket{a \oplus 2^{i t_i} b}\ket{b}$ to uncompute the input register (using the fact that the shifted-out bits were zero).
    Both of those out-of-place operations can be implemented in $O(m)$ gates, $O(m)$ space, and $O(\log m)$ depth, by using a tree of controlled-NOT gates to create $m$ copies of the control bit and then performing two layers of $m$ Toffoli gates between the control, input, and output registers, separated by a layer of NOT gates on the controls: the first layer XORs the output register by the shifted value of the input if the control is on, and the second layer XORs the output by the unshifted value if the control is off.
    Finally, the $m$ copies of the control bit are uncomputed by another tree of $O(m)$ controlled-NOT gates.

    Step~\ref{algostep:jacobi:setout} of Algorithm~\ref{algo:jacobi} can be implemented in $O(1)$ gates, depth, and space, because it only depends on a constant number of the bits of $t$ and $x'$.
    Step~\ref{algostep:jacobi:matchbits} of Algorithm~\ref{algo:jacobi} requires calling Algorithm~\ref{algo:match_bits_of_N}, which by Lemma~\ref{lemma:match_bits_cost} can be performed with the complexities specified in the Theorem.
    Step~\ref{algostep:jacobi:msize} consists of the computation of the Jacobi symbol of two $m$-bit inputs, which can be performed in $G_\mathrm{Jac}(m)$, space $S_\mathrm{Jac}(m)$, and depth $D_\mathrm{Jac}(m)$ by supposition.
    Finally, step~\ref{algostep:jacobi:uncompute} can be performed with the stated complexities given that steps~\ref{algostep:jacobi:matchbits} and~\ref{algostep:jacobi:makeodd} can.

    Thus all steps can be implemented reversibly with the stated complexities, completing the proof.
\end{proof}

\begin{lemma}\label{lemma:match_bits_cost}
Suppose there exists a quantum multiplication circuit on $t$-bit inputs with gates $G_\mathrm{mult}(t)$, space $S_\mathrm{mult}(t)$, and depth $D_\mathrm{mult}(t)$. 
Then, there exists a quantum circuit implementing Algorithm~\ref{algo:match_bits_of_N} with gates $O\left(\frac{n}{m} \cdot G_\mathrm{mult}(m)\right)$, space $O(S_\mathrm{mult}(m))$ qubits, and depth $O\left((\frac{n}{m} + \log m) \cdot D_\mathrm{mult}(m)\right)$.
\end{lemma}

\begin{proof}
    We proceed by showing that each step of the algorithm can be performed reversibly with the stated complexity.
    We note that $S_\mathrm{mult}(t) \geq O(m)$ because its inputs are quantum, so the overall algorithm's space is lower bounded by $O(m)$; and $D_\mathrm{mult}(t) \geq O(\log m)$ because each bit of a multiplier's input can affect $O(m)$ bits of its output, so the overall depth is lower bounded by $O(\log m)$.
    Both bounds hold for any choice of (reversible) multiplier.

    Both parts of Step~\ref{algostep:match_bits_of_N:precomputation} (and its uncomputation, Step~\ref{algostep:match_bits_of_N:unprecomputation}) are arithmetic divisions.
    When implemented via Newton iteration, the gate and space complexity of division is the same as multiplication up to a constant factor~\cite{knuth_art_1998}; the depth complexity for $t$-bit inputs is bounded by $O(\log t \cdot D_\mathrm{mult}(t))$ (although the bound improves to $O(D_\mathrm{mult}(t))$ if $D_\mathrm{mult}(t) \geq \Omega(t^\epsilon)$ for any $\epsilon > 0$).
    For step~\ref{algostep:match_bits_of_N:loop}, each iteration of the loop consists of a constant number of additions, subtractions, and multiplications, all of size $O(m)$.
    The additions and subtractions can be implemented in gate count and space $O(m)$, and depth $O(\log m)$, via quantum carry-lookahead addition~\cite{draper_logarithmic-depth_2006}.
    The multiplications can be performed with gates $G_\mathrm{mult}(m)$, space $S_\mathrm{mult}(m)$, and depth $D_\mathrm{mult}(m)$ by supposition.
    The loop has $n/m-1$ iterations, so overall, step~\ref{algostep:match_bits_of_N:loop} can be implemented in $O\left(\frac{n}{m} \cdot G_\mathrm{mult}(m)\right)$ gates, $O(S_\mathrm{mult}(m))$ qubits of space, and $O\left(\frac{n}{m} \cdot D_\mathrm{mult}(m)\right)$ depth.
    Thus all of Algorithm~\ref{algo:match_bits_of_N} can be implemented in the stated depth, space, and gate count, completing the proof.
\end{proof}

\subsection{Implications: Factoring Certain Integers in Sublinear Space and Depth}\label{sec:jacobicorollaries}

\sloppy In this section, we instantiate Algorithms~\ref{algo:jacobi} and~\ref{algo:match_bits_of_N}.
Here we focus on asymptotic costs, leaving the optimization of circuits for practical problem sizes to future work.
For multiplication, we use a parallelized quantum circuit for Sch\"onhage-Strassen multiplication~\cite{schonhage1971fast}, by which the product of two $t$-bit quantum integers can be computed in gate count $\widetilde{O}(t)$ and depth $\mathsf{polylog}(t)$, using $\widetilde{O}(t)$ total qubits~\cite{nie_quantum_2023}.
For computing the Jacobi symbol of two inputs of size $t$, there exist classical algorithms with complexity $\widetilde{O}(t)$~\cite{schonhageeuc, thull1990uni,bach1996algorithmic, Moller2008OnSA}; by standard reversible circuit techniques these algorithms can be made into quantum circuits with gate count, depth, and qubit count all at most $\widetilde{O}(t)$~\cite{bennett_logical_1973, bennett_timespace_1989, levine_note_1990}.
The following corollary results directly from instantiating Algorithms~\ref{algo:jacobi} and~\ref{algo:match_bits_of_N} with these constructions.

\begin{corollary}[Compare with Corollary~\ref{cor:highlevelwithschonhage}]\label{corr:minimal_gates}
    There exists a quantum circuit for the unitary of Equation~\eqref{eq:jac-unitary} with gate count $\widetilde{O}(n)$, depth $\widetilde{O}(n/m + m)$, and space $\widetilde{O}(m)$ qubits.

    Consequently, by Theorem~\ref{thm:highlevel}, we can recover prime $P$ and $Q$ (with $Q < 2^m$) from an $n$-bit input $N = P^2Q$ with $\widetilde{O}(n)$ gates in $\widetilde{O}(n/m + m)$ depth, using $\widetilde{O}(m)$ qubits.
\end{corollary}

\begin{remark}[Near-optimal parallelism for small $Q$]
    It is possible to achieve depth $\widetilde{O}(n/m + \log Q)$ using $\widetilde{O}(m)$ qubits for any $m \geq \log Q$, by using block size $m$ in Algorithm~\ref{algo:jacobi} but implementing step~\ref{algostep:jacobi:msize} via a recursive call to Algorithm~\ref{algo:jacobi} with a smaller block size (and possibly further levels of recursion if needed).
    This optimization becomes relevant when $\log Q < O(\sqrt{n})$, such that the $n/m$ term in the depth could dominate.
    In that regime, this trick allows the depth to be reduced as low as $\widetilde{O}(\log Q)$ at the expense of qubit count $\widetilde{O}(n/\log Q)$ (by setting $m=n/\log Q$).
    In general, for a target depth $d$ where $\log Q \leq d \leq n/\log Q$, at the $i^\mathrm{th}$ level of recursion the block size $m_i$ should be set to $m_i = m_{i-1}/d$ (with $m_1 = n/d$)
    and the recursion stops when $m_i < d$.
    This construction yields a depth $\widetilde{O}(d)$ and qubit count $\widetilde{O}(n/d)$.
    The space-time product (qubit count times depth) is $\widetilde{O}(n)$, the same as the gate count, thus nearly achieving the asymptotically optimal limit for parallelism of $O(1)$ operations per qubit per time step (up to polylogarithmic factors).
\end{remark}

\paragraph{Our Factoring-Based Proof of Quantumness.} 
We now state the implications of Corollary \ref{corr:minimal_gates} when factoring numbers of the form $N = P^2Q$ with $\log Q = \widetilde{\Theta}((\log N)^{2/3})$ and $P, Q$ are prime. 
As summarized in Section~\ref{sec:cntprelims}, there are no known classical special-purpose factoring algorithms that perform better than the general number field sieve~\cite{pollard, llmp, blp} on integers of this form. Yet by Corollary~\ref{corr:minimal_gates}, it is possible to quantumly factor these integers in much less space and depth than would be required for generic integers~\cite{shor97, Regev23} or even generic squarefull integers~\cite{LPDS12}. 
Indeed, applying Corollary~\ref{corr:minimal_gates} to numbers of that form yields the following result:

\begin{corollary}\label{corr:factoringpoq}
    Consider $n$-bit numbers of the form $N = P^2Q$, where $P, Q$ are primes and $\log Q = \widetilde{\Theta}(n^{2/3})$. 
    There exist quantum circuits for recovering $P$ and $Q$ from $N$ with $\widetilde{O}(n)$ gates, $\widetilde{O}(n^{2/3})$ depth, and $\widetilde{O}(n^{2/3})$ qubits.
    \label{corr:fix_m_results_in_terms_of_n}
\end{corollary}

\section{Completely Factoring Integers with Distinct Exponents in their Prime Factorization}\label{sec:fullyfactoring}

Here, we provide a black-box reduction that shows that any algorithm achieving the guarantees of Theorem~\ref{thm:highlevel} can in fact be used to \emph{completely factor} integers of the form $N = p_1^{\alpha_1} \ldots p_r^{\alpha_r}$ with $\alpha_1, \ldots, \alpha_r$ positive and distinct.

\begin{definition}[\cite{Akta2017OnTN}]
    We say that an integer $N$ is \emph{special} if all the exponents in its prime factorization are distinct i.e. we can write $N = p_1^{\alpha_1} \ldots p_r^{\alpha_r}$ for distinct primes $p_1, \ldots, p_r$ and distinct positive integers $\alpha_1, \ldots, \alpha_r$.
\end{definition}
\begin{remark}
    One might wonder whether special integers turn out to be classically easy to factor as well. We state some evidence that this is not the case here. Indeed, the density of these integers was studied by Akta\c{s} and Murty~\cite{Akta2017OnTN}, who showed that for any integer $N_\mathrm{max}$, the number of special integers $N \in [1, N_\mathrm{max}]$ is $\approx \frac{1.7 N_\mathrm{max}}{\ln N_\mathrm{max}}$. The following classes of integers that are classically even slightly easier than general to completely factor do not contain enough elements in total to cover all special numbers:
    \begin{itemize}
        \item Prime numbers: there are $\approx \frac{N_\mathrm{max}}{\ln N_\mathrm{max}}$ of these by the prime number theorem.
        \item Integers of the form $a^b$ for $b > 1$: there are at most $\widetilde{O}\left(\sqrt{N_\mathrm{max}}\right)$ of these.
        \item Sub-exponentially smooth integers (i.e. integers whose largest prime divisor is at most say $\exp\left(O(\left(\log N_\mathrm{max}\right)^{0.99}\right)$): there are $N_\mathrm{max} \cdot \exp\left(-\widetilde{O}(\left(\log N_\mathrm{max}\right)^{0.01}\right)$ of these~\cite{granville}.
    \end{itemize}
\end{remark}
With this in mind, we now turn our attention to completely factoring special integers.
The results and ideas in this section bear some high-level similarity to previous work by Yun~\cite{yun} in the context of factoring polynomials. Concretely, Yun shows that any polynomial $f(x)$ can be decomposed into a factorization $f(x) = g_1(x)^{\alpha_1}\ldots g_r(x)^{\alpha_r}$ where the $g_i$'s are squarefree and pairwise coprime and the $\alpha_i$'s are distinct. The starting point of Yun's algorithm is the observation that if $f$ is divisible by the square of some polynomial $g$, then $g$ will divide $\gcd(f, f')$. Similarly, in this section we will start from an algorithm for calculating squarefree decompositions of integers and obtain an algorithm for fully factoring special integers.

We first begin with a simple lemma. At a high level, we want to show that the $\Omega(1)$ success probability in Theorem~\ref{thm:highlevel} can be boosted to be very close to 1. This is not obvious since given $N = A^2B$ with $B$ squarefree, it may not be possible to efficiently determine whether the algorithm has succeeded in recovering $B$. We show that this is not difficult to work around.

\begin{lemma}\label{lemma:highlevelboosting}
    Let $N$ be a positive integer, with unique representation as $N = A^2B$ for $B$ squarefree. Moreover, we say that $N$ is \emph{very good} if it is composite and neither squarefree nor a square.
    
    Assume there exists an algorithm $\mathcal{A}$ that given a very good integer $N$, outputs either $B$ or a prime dividing $N$ with probability $\Omega(1)$.

    Then for any positive integer $T$, there exists another algorithm $\mathcal{A'}$ that given a positive integer $N = A^2B$ with $B$ squarefree, either outputs $B$ or a prime divisor of $N$ with probability $1 - \exp(-\Omega(T))$. This algorithm makes at most $T$ calls to $\mathcal{A}$ with the same input $N$. Outside of calls made to $\mathcal{A}$, the algorithm is classical and runs in time $\poly(\log N)$.
\end{lemma}
\begin{proof}
    First, we state the main idea. Suppose that $\mathcal{A}$ produces some composite $B'$ as output. The main observation is that while we cannot efficiently check whether $B' = B$, we can efficiently check that $N/B'$ is a square. Moreover, if $B'$ satisfies this condition then $B'$ must be divisible by $B$. With this in mind, $\mathcal{A'}$ will proceed as follows. We will use $B^*$ to denote the algorithm's final output:
    \begin{enumerate}
        \item If $N$ is prime, we can output $N$ itself and terminate. If $N$ is a square, we can easily check this and output $B^* = 1$ and terminate. Henceforth we can assume that $N$ is either very good or squarefree.
        \item Now run $\mathcal{A}$ $T$ times and let the outputs be $B_1, \ldots, B_T$. If there exists some $j$ such that $B_j$ is a prime divisor of $N$, output $B_j$ and terminate.
        \item Otherwise, let $S \subseteq \left\{B_1, \ldots, B_T\right\}$ be the set of all values $B'$ in this list such that $B'$ divides $N$ and $N/B'$ is a square.
        \item If $S = \emptyset$, output $B^* = N$ (this is equivalent to declaring that $N$ is squarefree). Otherwise, output $B^*$ as the smallest element in $S$ and terminate.
    \end{enumerate}
    First, suppose $N$ is very good. In this case, at least one of the runs of $\mathcal{A}$ will be successful (i.e. it outputs $B$ or a prime divisor of $N$) with probability $1 - \exp(-\Omega(T))$. Then, assuming at least one of the runs of $\mathcal{A}$ is successful, we have two cases:
    \begin{itemize}
        \item If the successful run produced a prime divisor, $\mathcal{A'}$ will detect this and output accordingly.
        \item If the successful run produced $B$, then this will be included in $S$. Moreover, all elements of $S$ must be divisible by $B$ (and hence $\geq B$). Hence taking the minimal element in $S$ will output $B$.
    \end{itemize}
    Finally, suppose $N$ is squarefree. In this case, we have no guarantee on the behavior of $\mathcal{A}$. But if it produces a prime divisor of $N$, $\mathcal{A'}$ will detect and output this. Otherwise, note that the only integer that could be included in $S$ is $N$ itself. So either $S$ will be empty or its smallest element will be $N$, and in either case $\mathcal{A'}$ will output $N$. The conclusion follows.
\end{proof}
The below theorem and its proof bear some high-level similarity to a result by Yun~\cite{yun} that shows that a similar factorization can easily be carried out for polynomials.
\begin{theorem}\label{thm:fullyfactoring}
    Assume there exists an algorithm $\mathcal{A}$ that given a positive integer $N = A^2B$ with $B$ squarefree and parameter $T$, either outputs $B$ or a prime divisor of $N$ with probability with $1 - \exp(-\Omega(T))$.

    Then there exists another algorithm $\mathcal{B}$ that, given a special integer $N$ as input, recovers the complete prime factorization of $N$ with probability $1 - \negl(\log N)$. This algorithm makes at most $O(\sqrt{\log N})$ calls to $\mathcal{A}$ with inputs $N'$ that are always $\leq N$ and with repetition parameter $T = \omega(\log N)$. Outside of calls made to $\mathcal{A}$, the algorithm is classical and runs in time $\poly(\log N)$.
\end{theorem}
\begin{proof}
    Let us write $N = p_1^{\alpha_1} \ldots p_r^{\alpha_r}$. Then note firstly that since the $\alpha_i$'s are distinct, we have $N \geq 2^{\alpha_1 + \ldots + \alpha_r} \geq 2^{\Omega(r^2)} \Rightarrow r \leq O(\sqrt{\log N})$. It hence suffices to show that we can accomplish the desired task with $O(r)$ calls.

    We present our algorithm in Algorithm~\ref{algo:specialintegers}. The efficiency is clear since after every call to $\mathcal{A}$, the number of distinct prime divisors of $M$ decreases by 1.
    As for correctness, note firstly that our procedure clearly preserves the fact that $M$ is special at each step. Updating $M \gets \sqrt{M}$ will halve all the exponents in its prime factorization which keeps them distinct. Otherwise, we take a prime and remove as many factors of it from $M$ as possible. This effectively just removes an element from the set of nonzero exponents in the prime factorization of $M$, which clearly preserves distinctness.

    It then remains to justify that if $M$ is special, then with all but negligible probability $B/\gcd(k, B)$ will be prime. Write $M = \prod_{i = 1}^s q_i^{\beta_i}$ for distinct primes $q_i$ and distinct positive integers $\beta_i$. Then if the output $B$ of $\mathcal{A}$ is not prime, we will have (with probability $1 - \negl(\log N)$)
    that $$B = \prod_{i \in [s]: \beta_i\text{ odd}} q_i.$$
    Now among the indices $i \in [s]$ such that $\beta_i$ is odd, let $i^*$ be the index such that $\beta_i$ is minimal. Then $\gamma = \beta_{i^*}$ (where $\gamma$ is defined as computed in Algorithm~\ref{algo:specialintegers}), and hence $$k = \left(\prod_{i \in [s]: \beta_i\text{ even}} q_i^{\beta_i}\right) \cdot \left(\prod_{i \in [s]: \beta_i\text{ odd}} q_i^{\beta_i - \beta_{i^*}}\right).$$
    The crucial point is that for any $i \in [s]$ with $i \neq i^*$ we will have $\beta_i \neq \beta_{i^*}$, because $M$ is special. In particular, for $i \in [s]$ such that $\beta_i$ is odd and $i \neq i^*$, $k$ must be divisible by $q_i$. On the other hand, $k$ is clearly not divisible by $q_{i^*}$. It follows that $$\gcd(k, B) = \prod_{i \in [s]: \beta_i\text{ odd and }i \neq i^*} q_i \Rightarrow \frac{B}{\gcd(k, B)} = q_{i^*},$$which is indeed prime. This completes our proof of the theorem.

    \begin{algorithm}
    \KwData{A special positive integer $N$.}
    \KwResult{A full factorization of $N$ (with probability $1 - \negl(\log N)$).}
    \begin{enumerate}
        \item Initialize $M$ to be $N$, and initialize $F$ to be the ``empty factorization'' (i.e. the factorization of 1). We will maintain the invariants that $M$ is a special divisor of $N$ and $F$ is the factorization of $N/M$.
        \item Repeat the following until $M = 1$:
        \begin{enumerate}
            \item If $M$ is a prime or prime power, add the prime factorization of $M$ to $F$, and update $M \gets 1$ (it is well-known that this can be efficiently done classically; we sketch this in Section~\ref{sec:cntprelims}).
            \item Else if $M$ is square, recurse, calling Algorithm~\ref{algo:specialintegers} on input $\sqrt{M}$; add two entries to $F$ for each prime factor in the result. Then set $M \gets 1$. 
            \item Otherwise, if neither the conditions in (a) nor (b) hold, apply algorithm $\mathcal{A}$ to $M$, with $T = \omega(\log N)$ so that the success probability is $1 - \negl(\log N)$. Now proceed as follows:
            \begin{itemize}
                \item If the output $B$ is prime: repeatedly divide $M$ by $B$ until $M$ is not divisible by $B$. Update $F$ accordingly and continue to the next step of the loop.
                \item Otherwise, we can assume that $B$ is squarefree and $M/B$ is square (with probability $1 - \negl(\log N)$). Then by repeatedly dividing $M$ by $B$, we can find integers $k, \gamma$ such that $M = k \cdot B^\gamma$ and $k$ is not divisible by $B$.

                Then compute $p = B/\gcd(k, B)$ (which can be done efficiently; see Section~\ref{sec:cntprelims} for an overview of some algorithms for computing GCDs) and check whether $p$ is prime. If it is not, abort (we will show that this almost never occurs). Otherwise, divide $M$ by as many factors of $p$ as possible and update $F$ accordingly.
            \end{itemize}
        \end{enumerate}
        \item Output $F$.
    \end{enumerate}
    \caption{Completely factoring special integers (see Theorem~\ref{thm:fullyfactoring})}\label{algo:specialintegers}
\end{algorithm}
\end{proof}
\noindent
Combining Corollary~\ref{cor:highlevelwithschonhage}, Lemma~\ref{lemma:highlevelboosting}, and Theorem~\ref{thm:fullyfactoring} yields the following result:
\begin{corollary}
    A special integer $N$ can be completely factored with success probability $1 - \negl(\log N)$ using $\omega((\log N)^{3/2})$ calls to a quantum circuit of size $\widetilde{O}(n)$.
\end{corollary}
\begin{proof}
    The circuit in Corollary~\ref{cor:highlevelwithschonhage} only requires $\widetilde{O}(n)$ gates. Then the algorithm in Lemma~\ref{lemma:highlevelboosting} can be realized with $T = \omega(\log N)$ calls to the circuit in Corollary~\ref{cor:highlevelwithschonhage} (here, the choice of $T$ is specified by Theorem~\ref{thm:fullyfactoring}.) Finally, the algorithm in Theorem~\ref{thm:fullyfactoring} can be realized with $O(\sqrt{\log N})$ calls to the algorithm of Lemma~\ref{lemma:highlevelboosting}. Putting these together, the conclusion follows.
\end{proof}

\ifanon
\else
\condparagraph{Acknowledgements.} The authors would like to thank Henry Corrigan-Gibbs for giving a stimulating talk at CRYPTO 2024 that inspired the beginning of this project, and for useful subsequent discussions.
The authors would also like to thank Isaac Chuang, Antoine Joux, Mikhail Lukin, and Peter Shor for insightful discussions. The authors would also like to thank Daniel J. Bernstein, Martin Eker{\aa}, Laura Lewis, Oded Regev, and anonymous reviewers for helpful comments and feedback on the manuscript.
\fi

\bibliographystyle{alpha}
\bibliography{main}

@article{LPDS12,
    author = "Jun Li and Xinhua Peng and Jiangfeng Du and Dieter Suter",
    title = "An Efficient Exact Quantum Algorithm for the Integer Square-free Decomposition Problem",
    year = "2012",
    journal = "Scientific Reports",
    volume = "2",
    issue = "202"
}

@article{schonhageeuc,
  title = {Schnelle {B}erechnung von {K}ettenbruchentwicklungen},
  volume = {1},
  ISSN = {1432-0525},
  url = {http://dx.doi.org/10.1007/BF00289520},
  DOI = {10.1007/bf00289520},
  number = {2},
  journal = {Acta Informatica},
  publisher = {Springer Science and Business Media LLC},
  author = {Sch{\"o}nhage,  A.},
  year = {1971},
  pages = {139–144}
}

@article{Moller2008OnSA,
  title={On {S}ch{\"o}nhage's algorithm and subquadratic integer gcd computation},
  author={Niels M{\"o}ller},
  journal={Math. Comput.},
  year={2008},
  volume={77},
  pages={589-607},
  url={https://api.semanticscholar.org/CorpusID:15065942}
}

@article{Regev23,
author = {Regev, Oded},
title = {An Efficient Quantum Factoring Algorithm},
year = {2025},
issue_date = {February 2025},
publisher = {Association for Computing Machinery},
address = {New York, NY, USA},
volume = {72},
number = {1},
issn = {0004-5411},
url = {https://doi.org/10.1145/3708471},
doi = {10.1145/3708471},
journal = {J. ACM},
month = jan,
articleno = {10},
numpages = {13}
}

@article{Montgomery1985,
  title = {Modular multiplication without trial division},
  volume = {44},
  ISSN = {1088-6842},
  url = {http://dx.doi.org/10.1090/S0025-5718-1985-0777282-X},
  DOI = {10.1090/s0025-5718-1985-0777282-x},
  number = {170},
  journal = {Mathematics of Computation},
  publisher = {American Mathematical Society (AMS)},
  author = {Montgomery,  Peter L.},
  year = {1985},
  pages = {519–521}
}

@inproceedings{yun,
author = {Yun, David Y.Y.},
title = {On square-free decomposition algorithms},
year = {1976},
isbn = {9781450377904},
publisher = {Association for Computing Machinery},
address = {New York, NY, USA},
url = {https://doi.org/10.1145/800205.806320},
doi = {10.1145/800205.806320},
booktitle = {Proceedings of the Third ACM Symposium on Symbolic and Algebraic Computation},
pages = {26–35},
numpages = {10},
location = {Yorktown Heights, New York, USA},
series = {SYMSAC '76}
}

@misc{thull1990uni,
  title={A Uni ed Approach to HGCD Algorithms for polynomials and integers},
  author={Thull, Klaus and Yap, Chee K},
  year={1990},
  publisher={Citeseer}
}

@inproceedings{DBLP:conf/focs/HalesH00,
  author       = {Lisa Hales and
                  Sean Hallgren},
  title        = {An Improved Quantum {Fourier} Transform Algorithm and Applications},
  booktitle    = {41st Annual Symposium on Foundations of Computer Science, {FOCS} 2000,
                  12-14 November 2000, Redondo Beach, California, {USA}},
  pages        = {515--525},
  publisher    = {{IEEE} Computer Society},
  year         = {2000},
  url          = {https://doi.org/10.1109/SFCS.2000.892139},
  doi          = {10.1109/SFCS.2000.892139},
  bibsource    = {dblp computer science bibliography, https://dblp.org}
}

@article{cemm,
    author = "Cleve, Richard and Ekert, Artur and Macchiavello, Chiara and Mosca, Michele",
    title = "{Quantum algorithms revisited}",
    eprint = "quant-ph/9708016",
    archivePrefix = "arXiv",
    doi = "10.1098/rspa.1998.0164",
    journal = "Proc. Roy. Soc. Lond. A",
    volume = "454",
    pages = "339",
    year = "1998"
}

@book{hardy75,
  author = {Hardy, G. H. and Wright, E. M.},
  edition = {Fourth},
  publisher = {Oxford},
  title = {An Introduction to the Theory of Numbers},
  year = 1975
}

@misc{conrad,
  author        = {Keith Conrad},
  title         = {{G}auss and {J}acobi sums on finite fields and $\mathbb{Z}/m\mathbb{Z}$},
  howpublished = "\url{http://kconrad.math.uconn.edu/blurbs/gradnumthy/Gauss-Jacobi-sums.pdf}"
}

@article{DBLP:journals/iacr/Chen24,
  author       = {Yilei Chen},
  title        = {Quantum Algorithms for Lattice Problems},
  journal      = {{IACR} Cryptol. ePrint Arch.},
  pages        = {555},
  year         = {2024},
  url          = {https://eprint.iacr.org/2024/555},
  bibsource    = {dblp computer science bibliography, https://dblp.org}
}

@book{bach1996algorithmic,
  title={Algorithmic Number Theory: Efficient algorithms},
  author={Bach, E. and Shallit, J.O.},
  number={v. 1},
  isbn={9780262024051},
  lccn={95025458},
  series={Algorithmic Number Theory},
  url={https://books.google.com/books?id=iJx1lP9ZcIkC},
  year={1996},
  publisher={MIT Press}
}

@article{takahashi,
  title={A quantum circuit for {Shor's} factoring algorithm using 2n+2 qubits},
  author={Takahashi, Yasuhiro and Kunihiro, Noboru},
  journal={Quantum Information \& Computation},
  volume={6},
  number={2},
  pages={184--192},
  year={2006},
  publisher={Rinton Press, Incorporated Paramus, NJ}
}

@article{KCVY21,
  author    = {Gregory D. Kahanamoku{-}Meyer and
               Soonwon Choi and
               Umesh V. Vazirani and
               Norman Y. Yao},
  title     = {Classically-Verifiable Quantum Advantage from a Computational {B}ell
               Test},
  journal   = {CoRR},
  volume    = {abs/2104.00687},
  year      = {2021},
  url       = {https://arxiv.org/abs/2104.00687},
  eprinttype = {arXiv},
  eprint    = {2104.00687},
  bibsource = {dblp computer science bibliography, https://dblp.org}
}

@misc{kahanamokumeyer2024fast,
      title={Fast quantum integer multiplication with zero ancillas}, 
      author={Gregory D. Kahanamoku-Meyer and Norman Y. Yao},
      year={2024},
      eprint={2403.18006},
      archivePrefix={arXiv},
      primaryClass={quant-ph}
}

@article{hrs17,
  author       = {Thomas H{\"{a}}ner and
                  Martin Roetteler and
                  Krysta M. Svore},
  title        = {Factoring using $2n+2$ qubits with {Toffoli}
                  based modular multiplication},
  journal      = {Quantum Inf. Comput.},
  volume       = {17},
  number       = {7{\&}8},
  pages        = {673--684},
  year         = {2017},
  url          = {https://doi.org/10.26421/QIC17.7-8-7},
  doi          = {10.26421/QIC17.7-8-7},
  bibsource    = {dblp computer science bibliography, https://dblp.org}
}

@inproceedings{DBLP:conf/pqcrypto/EkeraH17,
  author       = {Martin Eker{\aa} and
                  Johan H{\aa}stad},
  editor       = {Tanja Lange and
                  Tsuyoshi Takagi},
  title        = {Quantum Algorithms for Computing Short Discrete Logarithms and Factoring
                  {RSA} Integers},
  booktitle    = {Post-Quantum Cryptography - 8th International Workshop, PQCrypto 2017,
                  Utrecht, The Netherlands, June 26-28, 2017, Proceedings},
  series       = {Lecture Notes in Computer Science},
  volume       = {10346},
  pages        = {347--363},
  publisher    = {Springer},
  year         = {2017},
  url          = {https://doi.org/10.1007/978-3-319-59879-6\_20},
  doi          = {10.1007/978-3-319-59879-6\_20},
  bibsource    = {dblp computer science bibliography, https://dblp.org}
}

@article{DBLP:journals/quantum/GidneyE21,
  author       = {Craig Gidney and
                  Martin Eker{\aa}},
  title        = {How to factor 2048 bit {RSA} integers in 8 hours using 20 million
                  noisy qubits},
  journal      = {Quantum},
  volume       = {5},
  pages        = {433},
  year         = {2021},
  url          = {https://doi.org/10.22331/q-2021-04-15-433},
  doi          = {10.22331/q-2021-04-15-433},
  bibsource    = {dblp computer science bibliography, https://dblp.org}
}

@article{beckman,
  title={Efficient networks for quantum factoring},
  author={Beckman, David and Chari, Amalavoyal N and Devabhaktuni, Srikrishna and Preskill, John},
  journal={Physical Review A},
  volume={54},
  number={2},
  pages={1034},
  year={1996},
  publisher={APS}
}

@article{vedral,
  title={Quantum networks for elementary arithmetic operations},
  author={Vedral, Vlatko and Barenco, Adriano and Ekert, Artur},
  journal={Physical Review A},
  volume={54},
  number={1},
  pages={147},
  year={1996},
  publisher={APS}
}

@inproceedings{seifert2001using,
  title={Using fewer qubits in {Shor}’s factorization algorithm via simultaneous {Diophantine} approximation},
  author={Seifert, Jean-Pierre},
  booktitle={Cryptographers’ Track at the RSA Conference},
  pages={319--327},
  year={2001},
  organization={Springer}
}

@misc{zalka2006shors,
      title={Shor's algorithm with fewer (pure) qubits}, 
      author={Christof Zalka},
      year={2006},
      eprint={quant-ph/0601097},
      archivePrefix={arXiv},
      primaryClass={quant-ph}
}

@article{ProosZalka03,
  author       = {John Proos and
                  Christof Zalka},
  title        = {Shor's discrete logarithm quantum algorithm for elliptic curves},
  journal      = {Quantum Inf. Comput.},
  volume       = {3},
  number       = {4},
  pages        = {317--344},
  year         = {2003},
  url          = {https://doi.org/10.26421/QIC3.4-3},
  doi          = {10.26421/QIC3.4-3},
  bibsource    = {dblp computer science bibliography, https://dblp.org}
}

@article{lenstraecm,
 ISSN = {0003486X, 19398980},
 URL = {http://www.jstor.org/stable/1971363},
 author = {H. W. Lenstra},
 journal = {Annals of Mathematics},
 number = {3},
 pages = {649--673},
 publisher = {[Annals of Mathematics, Trustees of Princeton University on Behalf of the Annals of Mathematics, Mathematics Department, Princeton University]},
 title = {Factoring Integers with Elliptic Curves},
 urldate = {2024-11-01},
 volume = {126},
 year = {1987}
}

@article{schonhage1971fast,
  title={Fast multiplication of large numbers},
  author={Sch{\"o}nhage, Arnold and Strassen, Volker},
  journal={Computing},
  volume={7},
  pages={281--292},
  year={1971},
  publisher={Springer}
}

@InProceedings{pollard,
author="Pollard, J. M.",
editor="Lenstra, Arjen K.
and Lenstra, Hendrik W.",
title="Factoring with cubic integers",
booktitle="The development of the number field sieve",
year="1993",
publisher="Springer Berlin Heidelberg",
address="Berlin, Heidelberg",
pages="4--10",
isbn="978-3-540-47892-8"
}

@inproceedings{llmp,
  author       = {Arjen K. Lenstra and
                  Hendrik W. {Lenstra Jr.} and
                  Mark S. Manasse and
                  John M. Pollard},
  editor       = {Harriet Ortiz},
  title        = {The Number Field Sieve},
  booktitle    = {Proceedings of the 22nd Annual {ACM} Symposium on Theory of Computing,
                  May 13-17, 1990, Baltimore, Maryland, {USA}},
  pages        = {564--572},
  publisher    = {{ACM}},
  year         = {1990},
  url          = {https://doi.org/10.1145/100216.100295},
  doi          = {10.1145/100216.100295},
  bibsource    = {dblp computer science bibliography, https://dblp.org}
}

@article{Agrawal2004,
  title = {{PRIMES} is in {P}},
  volume = {160},
  ISSN = {0003-486X},
  url = {http://dx.doi.org/10.4007/annals.2004.160.781},
  DOI = {10.4007/annals.2004.160.781},
  number = {2},
  journal = {Annals of Mathematics},
  publisher = {Annals of Mathematics},
  author = {Agrawal,  Manindra and Kayal,  Neeraj and Saxena,  Nitin},
  year = {2004},
  month = sep,
  pages = {781–793}
}

@InProceedings{blp,
author="Buhler, J. P.
and Lenstra, H. W.
and Pomerance, Carl",
editor="Lenstra, Arjen K.
and Lenstra, Hendrik W.",
title="Factoring integers with the number field sieve",
booktitle="The development of the number field sieve",
year="1993",
publisher="Springer Berlin Heidelberg",
address="Berlin, Heidelberg",
pages="50--94",
isbn="978-3-540-47892-8"
}

@article{regev09,
  author       = {Oded Regev},
  title        = {On lattices, learning with errors, random linear codes, and cryptography},
  journal      = {J. {ACM}},
  volume       = {56},
  number       = {6},
  pages        = {34:1--34:40},
  year         = {2009},
  url          = {https://doi.org/10.1145/1568318.1568324},
  doi          = {10.1145/1568318.1568324},
  bibsource    = {dblp computer science bibliography, https://dblp.org}
}

@InProceedings{ekeragartner,
author="Eker{\aa}, Martin
and G{\"a}rtner, Joel",
editor="Saarinen, Markku-Juhani
and Smith-Tone, Daniel",
title="Extending {R}egev's Factoring Algorithm to Compute Discrete Logarithms",
booktitle="Post-Quantum Cryptography",
year="2024",
publisher="Springer Nature Switzerland",
address="Cham",
pages="211--242",
isbn="978-3-031-62746-0"
}

@article{martinis2019,
  title = {Quantum supremacy using a programmable superconducting processor},
  volume = {574},
  ISSN = {1476-4687},
  url = {http://dx.doi.org/10.1038/s41586-019-1666-5},
  DOI = {10.1038/s41586-019-1666-5},
  number = {7779},
  journal = {Nature},
  publisher = {Springer Science and Business Media LLC},
  author = {Arute,  Frank and Arya,  Kunal and Babbush,  Ryan and Bacon,  Dave and Bardin,  Joseph C. and Barends,  Rami and Biswas,  Rupak and Boixo,  Sergio and Brandao,  Fernando G. S. L. and Buell,  David A. and Burkett,  Brian and Chen,  Yu and Chen,  Zijun and Chiaro,  Ben and Collins,  Roberto and Courtney,  William and Dunsworth,  Andrew and Farhi,  Edward and Foxen,  Brooks and Fowler,  Austin and Gidney,  Craig and Giustina,  Marissa and Graff,  Rob and Guerin,  Keith and Habegger,  Steve and Harrigan,  Matthew P. and Hartmann,  Michael J. and Ho,  Alan and Hoffmann,  Markus and Huang,  Trent and Humble,  Travis S. and Isakov,  Sergei V. and Jeffrey,  Evan and Jiang,  Zhang and Kafri,  Dvir and Kechedzhi,  Kostyantyn and Kelly,  Julian and Klimov,  Paul V. and Knysh,  Sergey and Korotkov,  Alexander and Kostritsa,  Fedor and Landhuis,  David and Lindmark,  Mike and Lucero,  Erik and Lyakh,  Dmitry and Mandrà,  Salvatore and McClean,  Jarrod R. and McEwen,  Matthew and Megrant,  Anthony and Mi,  Xiao and Michielsen,  Kristel and Mohseni,  Masoud and Mutus,  Josh and Naaman,  Ofer and Neeley,  Matthew and Neill,  Charles and Niu,  Murphy Yuezhen and Ostby,  Eric and Petukhov,  Andre and Platt,  John C. and Quintana,  Chris and Rieffel,  Eleanor G. and Roushan,  Pedram and Rubin,  Nicholas C. and Sank,  Daniel and Satzinger,  Kevin J. and Smelyanskiy,  Vadim and Sung,  Kevin J. and Trevithick,  Matthew D. and Vainsencher,  Amit and Villalonga,  Benjamin and White,  Theodore and Yao,  Z. Jamie and Yeh,  Ping and Zalcman,  Adam and Neven,  Hartmut and Martinis,  John M.},
  year = {2019},
  month = oct,
  pages = {505–510}
}

@misc{grover2002creating,
      title={Creating superpositions that correspond to efficiently integrable probability distributions}, 
      author={Lov Grover and Terry Rudolph},
      year={2002},
      eprint={quant-ph/0208112},
      archivePrefix={arXiv},
      primaryClass={quant-ph}
}

@article{gidney2017factoring,
  title={Factoring with $n+2$ clean qubits and $n-1$ dirty qubits},
  author={Gidney, Craig},
  journal={arXiv preprint arXiv:1706.07884},
  year={2017}
}

@article{Harvey21,
  doi = {10.4007/annals.2021.193.2.4},
  url = {https://doi.org/10.4007/annals.2021.193.2.4},
  year = {2021},
  month = mar,
  publisher = {Annals of Mathematics},
  volume = {193},
  number = {2},
  author = {David Harvey and Joris van der Hoeven},
  title = {Integer multiplication in time ${O}(n\log n)$},
  journal = {Annals of Mathematics}
}

@article{Cop02,
  title={An approximate {Fourier} transform useful in quantum factoring},
  author={Coppersmith, Don},
  journal={arXiv preprint quant-ph/0201067},
  year={2002}
}

@article{SCHMIDTSAMOA200679,
title = {A New {R}abin-type Trapdoor Permutation Equivalent to Factoring},
journal = {Electronic Notes in Theoretical Computer Science},
volume = {157},
number = {3},
pages = {79-94},
year = {2006},
note = {Proceedings of the First International Workshop on Security and Trust Management (STM 2005)},
issn = {1571-0661},
doi = {https://doi.org/10.1016/j.entcs.2005.09.039},
author = {Katja Schmidt-Samoa}
}

@inproceedings{DBLP:conf/crypto/Takagi98,
  author       = {Tsuyoshi Takagi},
  editor       = {Hugo Krawczyk},
  title        = {Fast {RSA}-Type Cryptosystem Modulo p\({}^{\mbox{k}}\)q},
  booktitle    = {Advances in Cryptology - {CRYPTO} '98, 18th Annual International Cryptology
                  Conference, Santa Barbara, California, USA, August 23-27, 1998, Proceedings},
  series       = {Lecture Notes in Computer Science},
  volume       = {1462},
  pages        = {318--326},
  publisher    = {Springer},
  year         = {1998},
  url          = {https://doi.org/10.1007/BFb0055738},
  doi          = {10.1007/BFB0055738},
  bibsource    = {dblp computer science bibliography, https://dblp.org}
}

@article{DBLP:journals/joc/PaulusT00,
  author       = {Sachar Paulus and
                  Tsuyoshi Takagi},
  title        = {A New Public-Key Cryptosystem over a Quadratic Order with Quadratic
                  Decryption Time},
  journal      = {J. Cryptol.},
  volume       = {13},
  number       = {2},
  pages        = {263--272},
  year         = {2000},
  url          = {https://doi.org/10.1007/s001459910010},
  doi          = {10.1007/S001459910010},
  bibsource    = {dblp computer science bibliography, https://dblp.org}
}

@ARTICLE{50373,
  author={Okamoto, T.},
  journal={IEEE Transactions on Information Theory}, 
  title={A fast signature scheme based on congruential polynomial operations}, 
  year={1990},
  volume={36},
  number={1},
  pages={47-53},
  doi={10.1109/18.50373}}

@inproceedings{OkamotoU98,
  author       = {Tatsuaki Okamoto and
                  Shigenori Uchiyama},
  editor       = {Kaisa Nyberg},
  title        = {A New Public-Key Cryptosystem as Secure as Factoring},
  booktitle    = {Advances in Cryptology - {EUROCRYPT} '98, International Conference
                  on the Theory and Application of Cryptographic Techniques, Espoo,
                  Finland, May 31 - June 4, 1998, Proceeding},
  series       = {Lecture Notes in Computer Science},
  volume       = {1403},
  pages        = {308--318},
  publisher    = {Springer},
  year         = {1998},
  url          = {https://doi.org/10.1007/BFb0054135},
  doi          = {10.1007/BFB0054135},
  bibsource    = {dblp computer science bibliography, https://dblp.org}
}

@article{Peralta1996FasterFO,
  title={Faster factoring of integers of a special form},
  author={Ren{\'e} Peralta and Eiji Okamoto},
  journal={IEICE Transactions on Fundamentals of Electronics, Communications and Computer Sciences},
  year={1996},
  volume={79},
  pages={489-493},
  url={https://api.semanticscholar.org/CorpusID:14646285}
}

@article{BuchmannLenstra,
author = "Johannes A. Buchmann and Hendrik W. {Lenstra Jr}",
title = "Approximating rings of integers in number fields",
journal = "Journal Theorie de Nombres Bordeaux", volume = "6",
issue = "2",
pages = "221--260", 
year = "1994"
}

@inproceedings{rv24,
  author       = {Seyoon Ragavan and
                  Vinod Vaikuntanathan},
  editor       = {Leonid Reyzin and
                  Douglas Stebila},
  title        = {Space-Efficient and Noise-Robust Quantum Factoring},
  booktitle    = {Advances in Cryptology - {CRYPTO} 2024 - 44th Annual International
                  Cryptology Conference, Santa Barbara, CA, USA, August 18-22, 2024,
                  Proceedings, Part {VI}},
  series       = {Lecture Notes in Computer Science},
  volume       = {14925},
  pages        = {107--140},
  publisher    = {Springer},
  year         = {2024},
  url          = {https://doi.org/10.1007/978-3-031-68391-6\_4},
  doi          = {10.1007/978-3-031-68391-6\_4},
  bibsource    = {dblp computer science bibliography, https://dblp.org}
}

@article{DBLP:journals/corr/abs-2405-14381,
  author       = {Martin Eker{\aa} and
                  Joel G{\"{a}}rtner},
  title        = {A high-level comparison of state-of-the-art quantum algorithms for
                  breaking asymmetric cryptography},
  journal      = {CoRR},
  volume       = {abs/2405.14381},
  year         = {2024},
  url          = {https://doi.org/10.48550/arXiv.2405.14381},
  doi          = {10.48550/ARXIV.2405.14381},
  eprinttype    = {arXiv},
  eprint       = {2405.14381},
  bibsource    = {dblp computer science bibliography, https://dblp.org}
}

@inproceedings{DBLP:conf/crypto/BonehL96,
  author       = {Dan Boneh and
                  Richard J. Lipton},
  editor       = {Neal Koblitz},
  title        = {Algorithms for Black-Box Fields and their Application to Cryptography
                  (Extended Abstract)},
  booktitle    = {Advances in Cryptology - {CRYPTO} '96, 16th Annual International Cryptology
                  Conference, Santa Barbara, California, USA, August 18-22, 1996, Proceedings},
  series       = {Lecture Notes in Computer Science},
  volume       = {1109},
  pages        = {283--297},
  publisher    = {Springer},
  year         = {1996},
  url          = {https://doi.org/10.1007/3-540-68697-5\_22},
  doi          = {10.1007/3-540-68697-5\_22},
  bibsource    = {dblp computer science bibliography, https://dblp.org}
}

@inproceedings{DBLP:conf/crypto/CorriganGibbsW24,
  author       = {Henry Corrigan{-}Gibbs and
                  David J. Wu},
  editor       = {Leonid Reyzin and
                  Douglas Stebila},
  title        = {The One-Wayness of Jacobi Signatures},
  booktitle    = {Advances in Cryptology - {CRYPTO} 2024 - 44th Annual International
                  Cryptology Conference, Santa Barbara, CA, USA, August 18-22, 2024,
                  Proceedings, Part {V}},
  series       = {Lecture Notes in Computer Science},
  volume       = {14924},
  pages        = {3--13},
  publisher    = {Springer},
  year         = {2024},
  url          = {https://doi.org/10.1007/978-3-031-68388-6\_1},
  doi          = {10.1007/978-3-031-68388-6\_1},
  bibsource    = {dblp computer science bibliography, https://dblp.org}
}

@article{BISSON2011815,
title = {Computing the endomorphism ring of an ordinary elliptic curve over a finite field},
journal = {Journal of Number Theory},
volume = {131},
number = {5},
pages = {815-831},
year = {2011},
note = {Elliptic Curve Cryptography},
issn = {0022-314X},
doi = {https://doi.org/10.1016/j.jnt.2009.11.003},
url = {https://www.sciencedirect.com/science/article/pii/S0022314X09002789},
author = {Gaetan Bisson and Andrew V. Sutherland}
}

@inproceedings{DBLP:conf/eurocrypt/CastagnosL09,
  author       = {Guilhem Castagnos and
                  Fabien Laguillaumie},
  editor       = {Antoine Joux},
  title        = {On the Security of Cryptosystems with Quadratic Decryption: The Nicest
                  Cryptanalysis},
  booktitle    = {Advances in Cryptology - {EUROCRYPT} 2009, 28th Annual International
                  Conference on the Theory and Applications of Cryptographic Techniques,
                  Cologne, Germany, April 26-30, 2009. Proceedings},
  series       = {Lecture Notes in Computer Science},
  volume       = {5479},
  pages        = {260--277},
  publisher    = {Springer},
  year         = {2009},
  url          = {https://doi.org/10.1007/978-3-642-01001-9\_15},
  doi          = {10.1007/978-3-642-01001-9\_15},
  bibsource    = {dblp computer science bibliography, https://dblp.org}
}

@inproceedings{DBLP:conf/asiacrypt/CastagnosJLN09,
  author       = {Guilhem Castagnos and
                  Antoine Joux and
                  Fabien Laguillaumie and
                  Phong Q. Nguyen},
  editor       = {Mitsuru Matsui},
  title        = {Factoring \emph{pq}\({}^{\mbox{2}}\) with Quadratic Forms: Nice Cryptanalyses},
  booktitle    = {Advances in Cryptology - {ASIACRYPT} 2009, 15th International Conference
                  on the Theory and Application of Cryptology and Information Security,
                  Tokyo, Japan, December 6-10, 2009. Proceedings},
  series       = {Lecture Notes in Computer Science},
  volume       = {5912},
  pages        = {469--486},
  publisher    = {Springer},
  year         = {2009},
  url          = {https://doi.org/10.1007/978-3-642-10366-7\_28},
  doi          = {10.1007/978-3-642-10366-7\_28},
  bibsource    = {dblp computer science bibliography, https://dblp.org}
}

@article{Harvey_2022,
   title={A deterministic algorithm for finding r-power divisors},
   volume={8},
   ISSN={2363-9555},
   url={http://dx.doi.org/10.1007/s40993-022-00387-w},
   DOI={10.1007/s40993-022-00387-w},
   number={4},
   journal={Research in Number Theory},
   publisher={Springer Science and Business Media LLC},
   author={Harvey, David and Hittmeir, Markus},
   year={2022},
   month=oct }

@incollection{DBLP:series/isc/May10,
  author       = {Alexander May},
  editor       = {Phong Q. Nguyen and
                  Brigitte Vall{\'{e}}e},
  title        = {Using {LLL}-Reduction for Solving {RSA} and Factorization Problems},
  booktitle    = {The {LLL} Algorithm - Survey and Applications},
  series       = {Information Security and Cryptography},
  pages        = {315--348},
  publisher    = {Springer},
  year         = {2010},
  url          = {https://doi.org/10.1007/978-3-642-02295-1\_10},
  doi          = {10.1007/978-3-642-02295-1\_10},
  bibsource    = {dblp computer science bibliography, https://dblp.org}
}

@inproceedings{DBLP:conf/ctrsa/CoronFRZ16,
  author       = {Jean{-}S{\'{e}}bastien Coron and
                  Jean{-}Charles Faug{\`{e}}re and
                  Gu{\'{e}}na{\"{e}}l Renault and
                  Rina Zeitoun},
  editor       = {Kazue Sako},
  title        = {Factoring N=p{\^{}}rq{\^{}}s for Large r and s},
  booktitle    = {Topics in Cryptology - {CT-RSA} 2016 - The Cryptographers' Track at
                  the {RSA} Conference 2016, San Francisco, CA, USA, February 29 - March
                  4, 2016, Proceedings},
  series       = {Lecture Notes in Computer Science},
  volume       = {9610},
  pages        = {448--464},
  publisher    = {Springer},
  year         = {2016},
  url          = {https://doi.org/10.1007/978-3-319-29485-8\_26},
  doi          = {10.1007/978-3-319-29485-8\_26},
  bibsource    = {dblp computer science bibliography, https://dblp.org}
}

@inproceedings{DBLP:conf/crypto/BonehDH99,
  author       = {Dan Boneh and
                  Glenn Durfee and
                  Nick Howgrave{-}Graham},
  editor       = {Michael J. Wiener},
  title        = {Factoring {N} = p\({}^{\mbox{r}}\)q for Large r},
  booktitle    = {Advances in Cryptology - {CRYPTO} '99, 19th Annual International Cryptology
                  Conference, Santa Barbara, California, USA, August 15-19, 1999, Proceedings},
  series       = {Lecture Notes in Computer Science},
  volume       = {1666},
  pages        = {326--337},
  publisher    = {Springer},
  year         = {1999},
  url          = {https://doi.org/10.1007/3-540-48405-1\_21},
  doi          = {10.1007/3-540-48405-1\_21},
  bibsource    = {dblp computer science bibliography, https://dblp.org}
}

@article{shor97,
  author       = {Peter W. Shor},
  title        = {Polynomial-Time Algorithms for Prime Factorization and Discrete Logarithms
                  on a Quantum Computer},
  journal      = {{SIAM} J. Comput.},
  volume       = {26},
  number       = {5},
  pages        = {1484--1509},
  year         = {1997},
  url          = {https://doi.org/10.1137/S0097539795293172},
  doi          = {10.1137/S0097539795293172},
  bibsource    = {dblp computer science bibliography, https://dblp.org}
}

@article{granville,
author = {Granville, Andrew},
year = {2000},
month = {01},
pages = {},
title = {Smooth numbers: Computational number theory and beyond},
volume = {44},
journal = {Math. Sci. Res. Inst. Publ.},
}

@article{Akta2017OnTN,
  title={On the number of special numbers},
  author={Kevser Akta\c{s} and M. Ram Murty},
  journal={Proceedings - Mathematical Sciences},
  year={2017},
  volume={127},
  pages={423-430},
  url={https://api.semanticscholar.org/CorpusID:125439724}
}

@article{gidney2019,
      title={Asymptotically efficient quantum {Karatsuba} multiplication},
  author={Gidney, Craig},
  journal={arXiv preprint arXiv:1904.07356},
  year={2019}
}

@inproceedings{mulder24,
    title = "FAST SQUARE-FREE DECOMPOSITION OF INTEGERS USING
CLASS GROUPS",
    author = "Erik Mulder",
    year = "2024",
    conference = "Algorithmic Number Theory Symposium"
}

@inproceedings{CleveW00,
  author       = {Richard Cleve and
                  John Watrous},
  title        = {Fast parallel circuits for the quantum {Fourier} transform},
  booktitle    = {41st Annual Symposium on Foundations of Computer Science, {FOCS} 2000,
                  12-14 November 2000, Redondo Beach, California, {USA}},
  pages        = {526--536},
  publisher    = {{IEEE} Computer Society},
  year         = {2000},
  url          = {https://doi.org/10.1109/SFCS.2000.892140},
  doi          = {10.1109/SFCS.2000.892140},
  bibsource    = {dblp computer science bibliography, https://dblp.org}
}

@article{beauregard,
  author       = {St{\'{e}}phane Beauregard},
  title        = {Circuit for {Shor's} algorithm using 2n+3 qubits},
  journal      = {Quantum Inf. Comput.},
  volume       = {3},
  number       = {2},
  pages        = {175--185},
  year         = {2003},
  url          = {https://doi.org/10.26421/QIC3.2-8},
  doi          = {10.26421/QIC3.2-8},
  bibsource    = {dblp computer science bibliography, https://dblp.org}
}

@inproceedings{roetteler17,
  title={Quantum resource estimates for computing elliptic curve discrete logarithms},
  author={Roetteler, Martin and Naehrig, Michael and Svore, Krysta M and Lauter, Kristin},
  booktitle={Advances in Cryptology--ASIACRYPT 2017: 23rd International Conference on the Theory and Applications of Cryptology and Information Security, Hong Kong, China, December 3-7, 2017, Proceedings, Part II 23},
  pages={241--270},
  year={2017},
  organization={Springer}
}

@article{levine_note_1990,
	title = {A {Note} on {Bennett}’s {Time}-{Space} {Tradeoff} for {Reversible} {Computation}},
	volume = {19},
	issn = {0097-5397},
	url = {https://epubs.siam.org/doi/abs/10.1137/0219046},
	doi = {10.1137/0219046},
	number = {4},
	urldate = {2020-10-09},
	journal = {SIAM Journal on Computing},
	author = {Levine, Robert Y. and Sherman, Alan T.},
	month = aug,
	year = {1990},
	note = {Publisher: Society for Industrial and Applied Mathematics},
	pages = {673--677},
}

@article{bennett_timespace_1989,
	title = {Time/{Space} {Trade}-{Offs} for {Reversible} {Computation}},
	volume = {18},
	issn = {0097-5397},
	url = {https://epubs.siam.org/doi/abs/10.1137/0218053},
	doi = {10.1137/0218053},
	number = {4},
	urldate = {2020-10-09},
	journal = {SIAM Journal on Computing},
	author = {Bennett, Charles H.},
	month = aug,
	year = {1989},
	note = {Publisher: Society for Industrial and Applied Mathematics},
	pages = {766--776},
}

@article{DBLP:journals/jacm/PippengerF79,
  author       = {Nicholas Pippenger and
                  Michael J. Fischer},
  title        = {Relations Among Complexity Measures},
  journal      = {J. {ACM}},
  volume       = {26},
  number       = {2},
  pages        = {361--381},
  year         = {1979},
  url          = {https://doi.org/10.1145/322123.322138},
  doi          = {10.1145/322123.322138},
  bibsource    = {dblp computer science bibliography, https://dblp.org}
}

@article{bennett_logical_1973,
	title = {Logical {Reversibility} of {Computation}},
	volume = {17},
	issn = {0018-8646},
	doi = {10.1147/rd.176.0525},
	number = {6},
	journal = {IBM Journal of Research and Development},
	author = {Bennett, C. H.},
	month = nov,
	year = {1973},
	note = {Conference Name: IBM Journal of Research and Development},
	pages = {525--532},
}

@book{knuth_art_1998,
  title = {The Art of Computer Programming, {{Volume II}}: {{Seminumerical Algorithms}}, 3rd {{Edition}}},
  author = {Knuth, Donald Ervin},
  year = {1998},
  publisher = {Addison-Wesley},
  isbn = {0-201-89684-2}
}

@article{draper_logarithmic-depth_2006,
  title = {A Logarithmic-Depth Quantum Carry-Lookahead Adder},
  author = {Draper, Thomas G. and Kutin, Samuel A. and Rains, Eric M. and Svore, Krysta M.},
  year = {2006},
  month = jul,
  journal = {Quantum Information \& Computation},
  volume = {6},
  number = {4},
  pages = {351--369},
  issn = {1533-7146}
}

@article{brakerski_cryptographic_2021,
  title = {A {{Cryptographic Test}} of {{Quantumness}} and {{Certifiable Randomness}} from a {{Single Quantum Device}}},
  author = {Brakerski, Zvika and Christiano, Paul and Mahadev, Urmila and Vazirani, Umesh and Vidick, Thomas},
  year = {2021},
  month = aug,
  journal = {Journal of the ACM (JACM)},
  publisher = {ACM},
  doi = {10.1145/3441309},
  urldate = {2022-06-23},
  langid = {english}
}

@inproceedings{brakerski_simpler_2020,
  title = {Simpler {{Proofs}} of {{Quantumness}}},
  booktitle = {15th {{Conference}} on the {{Theory}} of {{Quantum Computation}}, {{Communication}} and {{Cryptography}} ({{TQC}} 2020)},
  author = {Brakerski, Zvika and Koppula, Venkata and Vazirani, Umesh and Vidick, Thomas},
  editor = {Flammia, Steven T.},
  year = {2020},
  series = {Leibniz {{International Proceedings}} in {{Informatics}} ({{LIPIcs}})},
  volume = {158},
  pages = {8:1--8:14},
  publisher = {Schloss Dagstuhl--Leibniz-Zentrum f{\"u}r Informatik},
  address = {Dagstuhl, Germany},
  issn = {1868-8969},
  doi = {10.4230/LIPIcs.TQC.2020.8},
  urldate = {2022-06-23},
  isbn = {978-3-95977-146-7}
}

@inproceedings{yamakawa_verifiable_2022,
  title = {Verifiable {{Quantum Advantage}} without {{Structure}}},
  booktitle = {2022 {{IEEE}} 63rd {{Annual Symposium}} on {{Foundations}} of {{Computer Science}} ({{FOCS}})},
  author = {Yamakawa, Takashi and Zhandry, Mark},
  year = {2022},
  month = oct,
  pages = {69--74},
  issn = {2575-8454},
  doi = {10.1109/FOCS54457.2022.00014}
}

@inproceedings{kalai_quantum_2023,
  title = {Quantum {{Advantage}} from {{Any Non-local Game}}},
  booktitle = {Proceedings of the 55th {{Annual ACM Symposium}} on {{Theory}} of {{Computing}}},
  author = {Kalai, Yael and Lombardi, Alex and Vaikuntanathan, Vinod and Yang, Lisa},
  year = {2023},
  month = jun,
  series = {{{STOC}} 2023},
  pages = {1617--1628},
  publisher = {Association for Computing Machinery},
  address = {New York, NY, USA},
  doi = {10.1145/3564246.3585164},
  urldate = {2023-09-05},
  isbn = {978-1-4503-9913-5}
}

@misc{aaronson_verifiable_2024,
  title = {On Verifiable Quantum Advantage with Peaked Circuit Sampling},
  author = {Aaronson, Scott and Zhang, Yuxuan},
  year = {2024},
  month = apr,
  number = {arXiv:2404.14493},
  eprint = {2404.14493},
  primaryclass = {quant-ph},
  publisher = {arXiv},
  doi = {10.48550/arXiv.2404.14493},
  urldate = {2024-04-29},
  archiveprefix = {arXiv}
}

@article{canetti_random_2004,
  title = {The Random Oracle Methodology, Revisited},
  author = {Canetti, Ran and Goldreich, Oded and Halevi, Shai},
  year = {2004},
  month = jul,
  journal = {Journal of the ACM},
  volume = {51},
  number = {4},
  pages = {557--594},
  issn = {0004-5411},
  doi = {10.1145/1008731.1008734},
  urldate = {2023-03-15},
}

@article{koblitz_random_2015,
  title = {The Random Oracle Model: A Twenty-Year Retrospective},
  shorttitle = {The Random Oracle Model},
  author = {Koblitz, Neal and Menezes, Alfred J.},
  year = {2015},
  month = dec,
  journal = {Designs, Codes and Cryptography},
  volume = {77},
  number = {2},
  pages = {587--610},
  issn = {1573-7586},
  doi = {10.1007/s10623-015-0094-2},
  urldate = {2021-08-13},
  langid = {english},
}

@inproceedings{haner_improved_2020,
  title = {Improved {{Quantum Circuits}} for {{Elliptic Curve Discrete Logarithms}}},
  booktitle = {Post-{{Quantum Cryptography}}},
  author = {H{\"a}ner, Thomas and Jaques, Samuel and Naehrig, Michael and Roetteler, Martin and Soeken, Mathias},
  editor = {Ding, Jintai and Tillich, Jean-Pierre},
  year = {2020},
  series = {Lecture {{Notes}} in {{Computer Science}}},
  pages = {425--444},
  publisher = {Springer International Publishing},
  address = {Cham},
  doi = {10.1007/978-3-030-44223-1_23},
  isbn = {978-3-030-44223-1},
  langid = {english}
}

@article{nie_quantum_2023,
  title = {Quantum {{Circuit Design}} for {{Integer Multiplication Based}} on {{Sch{\"o}nhage}}--{{Strassen Algorithm}}},
  author = {Nie, Junhong and Zhu, Qinlin and Li, Meng and Sun, Xiaoming},
  year = {2023},
  month = dec,
  journal = {IEEE Transactions on Computer-Aided Design of Integrated Circuits and Systems},
  volume = {42},
  number = {12},
  pages = {4791--4802},
  issn = {1937-4151},
  doi = {10.1109/TCAD.2023.3279300},
  urldate = {2024-11-07}
}

@misc{miller_hidden-state_2024,
  title = {Hidden-{{State Proofs}} of {{Quantumness}}},
  author = {Miller, Carl A.},
  year = {2024},
  month = oct,
  number = {arXiv:2410.06368},
  eprint = {2410.06368},
  publisher = {arXiv},
  doi = {10.48550/arXiv.2410.06368},
  urldate = {2024-11-26},
  archiveprefix = {arXiv}
}

@article{alnawakhtha_lattice-based_2024,
  title = {Lattice-{{Based Quantum Advantage}} from {{Rotated Measurements}}},
  author = {Alnawakhtha, Yusuf and Mantri, Atul and Miller, Carl A. and Wang, Daochen},
  year = {2024},
  month = jul,
  journal = {Quantum},
  volume = {8},
  pages = {1399},
  publisher = {Verein zur F{\"o}rderung des Open Access Publizierens in den Quantenwissenschaften},
  doi = {10.22331/q-2024-07-04-1399},
  urldate = {2024-11-26},
  langid = {british}
}

@inproceedings{morimae_proofs_2023,
  title = {Proofs of Quantumness from Trapdoor Permutations},
  booktitle = {14th Innovations in Theoretical Computer Science Conference ({{ITCS}} 2023)},
  author = {Morimae, Tomoyuki and Yamakawa, Takashi},
  editor = {Tauman Kalai, Yael},
  year = {2023},
  series = {Leibniz International Proceedings in Informatics (Lipics)},
  volume = {251},
  pages = {87:1--87:14},
  publisher = {Schloss Dagstuhl -- Leibniz-Zentrum f{\"u}r Informatik},
  address = {Dagstuhl, Germany},
  issn = {1868-8969},
  doi = {10.4230/LIPIcs.ITCS.2023.87},
  isbn = {978-3-95977-263-1},
  urn = {urn:nbn:de:0030-drops-175900}
}

@misc{arabadjieva_single-round_2024,
  title = {Single-{{Round Proofs}} of {{Quantumness}} from {{Knowledge Assumptions}}},
  author = {Arabadjieva, Petia and Gheorghiu, Alexandru and Gitton, Victor and Metger, Tony},
  year = {2024},
  month = may,
  number = {arXiv:2405.15736},
  eprint = {2405.15736},
  primaryclass = {quant-ph},
  publisher = {arXiv},
  doi = {10.48550/arXiv.2405.15736},
  urldate = {2024-05-30},
  archiveprefix = {arXiv}
}

@inproceedings{bernstein_post-quantum_2017,
  title = {Post-Quantum {{RSA}}},
  booktitle = {Post-{{Quantum Cryptography}}},
  author = {Bernstein, Daniel J. and Heninger, Nadia and Lou, Paul and Valenta, Luke},
  editor = {Lange, Tanja and Takagi, Tsuyoshi},
  year = {2017},
  pages = {311--329},
  publisher = {Springer International Publishing},
  address = {Cham},
  doi = {10.1007/978-3-319-59879-6_18},
  isbn = {978-3-319-59879-6},
  langid = {english}
}

@inproceedings{bernstein_low-resource_2017,
  title = {A {{Low-Resource Quantum Factoring Algorithm}}},
  booktitle = {Post-{{Quantum Cryptography}}},
  author = {Bernstein, Daniel J. and Biasse, Jean-Fran{\c c}ois and Mosca, Michele},
  editor = {Lange, Tanja and Takagi, Tsuyoshi},
  year = {2017},
  pages = {330--346},
  publisher = {Springer International Publishing},
  address = {Cham},
  doi = {10.1007/978-3-319-59879-6_19},
  isbn = {978-3-319-59879-6},
  langid = {english}
}

@article{mosca_speeding_2020,
  title = {On Speeding up Factoring with Quantum {{SAT}} Solvers},
  author = {Mosca, Michele and Basso, Joao Marcos Vensi and Verschoor, Sebastian R.},
  year = {2020},
  month = sep,
  journal = {Scientific Reports},
  volume = {10},
  number = {1},
  pages = {15022},
  publisher = {Nature Publishing Group},
  issn = {2045-2322},
  doi = {10.1038/s41598-020-71654-y},
  urldate = {2024-12-13},
  copyright = {2020 The Author(s)},
  langid = {english}
}

@misc{chevignard_reducing_2024,
  title = {Reducing the {{Number}} of {{Qubits}} in {{Quantum Factoring}}},
  author = {Chevignard, Cl{\'e}mence and Fouque, Pierre-Alain and Schrottenloher, Andr{\'e}},
  year = {2024},
  number = {2024/222},
  urldate = {2025-03-27}
}

@article{DBLP:journals/iacr/Schanck18,
  author       = {John M. Schanck},
  title        = {Multi-power Post-quantum {RSA}},
  journal      = {{IACR} Cryptol. ePrint Arch.},
  pages        = {325},
  year         = {2018},
  url          = {https://eprint.iacr.org/2018/325},
  bibsource    = {dblp computer science bibliography, https://dblp.org}
}


\end{document}